\documentclass[final]{ws-idaqp}
\usepackage{graphicx}
\usepackage{graphicx}

\usepackage{amsmath}
\usepackage{amssymb}
\usepackage{amsfonts}
\usepackage{graphicx}
\usepackage{color}
\usepackage{datetime}
\usepackage{hyperref}

\usepackage{framed} % for pdf, bitmapped graphics files
\usepackage{graphicx} % for pdf, bitmapped graphics files
\usepackage{amsmath} % assumes amsmath package installed
\usepackage{amssymb}  % assumes amsmath package installed
\newcommand{\sinc}{\mathrm{sinc}}
\newcommand{\tanc}{\mathrm{tanc}}
\newcommand{\tanhc}{\mathrm{tanhc}}

\DeclareMathAlphabet{\bit}{OML}{cmm}{b}{it}
\def\<{\leqslant}           % nice less than or equal to sign
\def\>{\geqslant}           % nice larger than or equal to sign
\def\d{\partial}

\def\Re{\mathrm{Re}}   % real part
\def\Im{\mathrm{Im}}   % imaginary part
\def\rprod{\mathop{\overrightarrow{\prod}}}

\def\col{\mathrm{vec}}   % vectorization of matrices
\def\cH{\mathcal{H}}   % Hardy space
\def\mZ{{\mathbb Z}}    % set of integers
\def\mR{{\mathbb R}}    % real line
\def\mC{\mathbb{C}}    % complex plane

\def\Tr{\mathrm{Tr}}       % matrix trace
\def\rT{{\rm T}}        % matrix transpose
\def\rF{{\rm F}}        % matrix transpose
\def\diam{\diamond}       % matrix trace
\def\bTheta{\bit{\Theta}}
\def\bitJ{\bit{J}}
\def\bitI{\bit{I}}

\def\bP{\mathbf{P}}    % probability
\def\bE{\mathbf{E}}    % expectation

\def\[[[{[\![\![}
\def\]]]{]\!]\!]}

\def\re{{\rm e}}        % number e
\def\rd{{\rm d}}        % differential

\def\bJ{\mathbf{J}}

\def\br{\mathbf{r}}
\def\x{\times}
\def\ox{\otimes}
\def\op{\oplus}

\def\fF{\mathfrak{F}}
\def\fT{\mathfrak{T}}

\def\fH{\mathfrak{H}}
\def\fS{\mathfrak{S}}

\def\sF{\mathsf{F}}

\def\sL{\mathsf{L}}
\def\sK{\mathsf{K}}
\def\sV{\mathsf{V}}

\def\cF{\mathcal{F}}

\def\cD{\mathcal{D}}

\def\cM{\mathcal{M}}

\def\cU{\mathcal{ U}}
\def\cC{\mathcal{ C}}
\def\cR{\mathcal{ R}}

\def\cI{\mathcal{I}}
\def\cP{\mathcal{P}}

\def\cA{\mathcal{ A}}
\def\cB{\mathcal{ B}}
\def\cE{\mathcal{ E}}

\def\cS{{\mathcal S}}
\def\cT{\mathcal{T}}

\def\mT{{\mathbb T}}
\def\mZ{{\mathbb Z}}

\def\veps{\varepsilon}
\def\eps{\epsilon}
\def\ups{\upsilon}
\def\Ups{\Upsilon}

\begin{document}
%\twocolumn
\markboth{I.G.Vladimirov, I.R.Petersen}
{Risk-sensitive criteria for translation invariant linear quantum stochastic  networks}

%%%%%%%%%%%%%%%%%%% Publisher's Area please ignore %%%%%%%%%%%%%%%%%%%%%%
\catchline{}{}{}{}{}
%%%%%%%%%%%%%%%%%%%%%%%%%%%%%%%%%%%%%%%%%%%%%%%%%%%%%%%%%%%%%%%%%%%%%%%%%

\title{\Large Infinite-horizon risk-sensitive performance criteria for translation invariant networks of linear quantum stochastic systems}

\author{Igor G. Vladimirov\footnote{Corresponding author},
\qquad
Ian R. Petersen$^\dagger$}

\address{
School of Engineering,\\
College of Engineering and Computer Science,\\
Australian National University, Canberra,\\
Acton, ACT 2601, Australia\\
$^*$igor.g.vladimirov@gmail.com\\
$^\dagger$i.r.petersen@gmail.com
}

\maketitle

%\begin{history}
%\received{(Day Month Year)}
%\revised{(Day Month Year)}
%\published{(Day Month Year)}
%\comby{(xxxxxxxxxx)}
%\end{history}
\thispagestyle{empty}
\begin{dedication}
\emph{Dedicated to the memory of Professor Robin Lyth Hudson}
\end{dedication}
\begin{abstract}
This paper is concerned with 
networks of identical linear quantum stochastic systems which interact with each other and external bosonic fields in a translation invariant fashion. The systems
are associated with sites of a multidimensional lattice and are governed by coupled linear
quantum stochastic differential equations (QSDEs). The block Toeplitz  coefficients of these QSDEs   are specified by the energy and coupling matrices which quantify the Hamiltonian and coupling operators for the component systems. We discuss the invariant Gaussian quantum state of the network when it satisfies a stability condition and is driven by statistically independent vacuum fields. A quadratic-exponential functional (QEF) is considered as a risk-sensitive performance criterion for a finite fragment of the network over a bounded time interval. This functional involves a quadratic function of dynamic variables of the component systems with a block Toeplitz weighting matrix.
Assuming the invariant state, we study the spatio-temporal asymptotic rate of the QEF per unit time and per lattice site in the thermodynamic limit of unboundedly growing time horizons and fragments of the lattice. A spatio-temporal frequency-domain formula is obtained for the QEF rate in terms of two spectral functions associated with the real and imaginary parts of the invariant quantum covariance kernel of the network variables. A homotopy method and asymptotic expansions for evaluating the QEF rate are also discussed.

\end{abstract}

\keywords{Linear quantum stochastic network;
translation invariance;
Gaussian quantum state;
quadratic-exponential functional;
spatio-temporal growth rate.
}

\ccode{AMS Subject Classification:
81S22, % Open systems, reduced dynamics, master equations, decoherence
81S25, % Quantum stochastic calculus
81P16, % Quantum state spaces, operational and probabilistic concepts
81R15,  % 	Operator algebra methods applied to problems in quantum theory
47B35,  % Toeplitz operators, Hankel operators, Wiener-Hopf operators [See also 45P05, 47G10 for other
%integral operators; see also 32A25, 32M15]
47L80,  %Algebras of specific types of operators (Toeplitz, integral, pseudodifferential, etc.)
15B05,  %Toeplitz, Cauchy, and related matrices
%81Q15, % Perturbation theories for operators and differential equations
%93C55,      % Discrete-time systems
%94A17,      % Measures of information, entropy
%93B05,      % Controllability
93E15,  	% Stochastic stability
%93E20,      % Optimal stochastic control
%60J05,  	% Discrete-time Markov processes on general state spaces
%49L20,  	% Dynamic programming method
%90C40,  	% Markov and semi-Markov decision processes
37L40,      % Invariant measures
60G15,   	% Gaussian processes
%93B35,      %Sensitivity (robustness)
93B51.      % Design techniques (robust design, computer-aided design, etc.)
%
%81S25,   	%Quantum stochastic calculus
%81S30       % Phase-space methods including Wigner distributions, etc.
%81S05,      %Canonical quantization, commutation relations and statistics
%81S22,       % Open systems, reduced dynamics, master equations, decoherence
%81P16,   	%Quantum state spaces, operational and probabilistic concepts
%81P40,   	%Quantum coherence, entanglement, quantum correlations
%81Q93,   	%Quantum control
%81Q10,   	%Selfadjoint operator theory in quantum theory, including spectral analysis
%60G15.   	% Gaussian processes
%\and
%93E20      % Optimal stochastic control
%22E46, 53C35, 57S20
}

\section{Introduction}
\label{intro}

Translational symmetries, which are  ubiquitous in the physical world, play an important role in collective properties of large-scale networks of interacting systems. For example, thermodynamic and mechanical characteristics of crystalline solids (including the heat capacity and speed of sound) are substantially affected by spatial periodicity in the arrangements   of atoms  in such states of matter and translation invariance of their mutual interaction, which is taken  into account by the phonon theory\cite{S_1990}.

Translation invariant interconnections  are also used in quantum metamaterials\cite{QSMGH_2011,RZSN_2008,Za_2011,Zh_2011}, where coupled identical quantum systems form one, two or three-dimensional periodic arrays\cite{Za_2012}. The resulting quantum composite system is  effectively homogeneous (in the sense of translational symmetries) on the scale of relevant wavelengths.  These artificial materials aim to unveil and exploit qualitatively new properties of light-matter interaction, such as in artificial crystals of atoms trapped at nodes of an optical lattice which can be controlled by external fields and used for entanglement generation\cite{CBFFFRCIP_2013} or as a quantum memory\cite{HCHJ13,NDMRLLWJ_2010,YJ14}. Similar architectures (in the form of one-dimensional chains) are present in cascaded quantum systems for generating pure Gaussian states\cite{KY12:pra,MWPY_2014,Y_2012}.

The present paper is concerned with networks of identical linear quantum stochastic systems\cite{NY_2017,P_2017}, or open quantum harmonic oscillators (OQHOs),   which interact with each other and external bosonic fields in a translation invariant fashion. The systems are associated with sites of a multidimensional lattice and are governed by coupled linear quantum stochastic differential equations (QSDEs) driven by quantum Wiener processes in the sense of the Hudson-Parthasarathy  calculus\cite{HP_1984,P_1992,P_2015}. In accordance with the translation invariance of the quantum network (with respect to the additive group structure of the lattice), the coefficients of these QSDEs   are organised as block Toeplitz matrices and are specified by the energy and coupling parameters which quantify the Hamiltonian and coupling operators for the component systems. This parameterization secures the fulfillment of physical realizability (PR) conditions, which extend those for OQHOs with a finite number of degrees of freedom\cite{JNP_2008,SP_2012}  and are similar to the network counterpart from Ref. \refcite{VP_2014} using the spatial Fourier transforms (SFTs).

We employ the homomorphism between the algebra of block Toeplitz matrices, the corresponding  convolution algebra of matrix-valued maps on the lattice and the algebra of SFTs with the pointwise multiplication over an appropriately dimensioned torus. This machinery represents system theoretic operations (such as concatenation and feedback interconnection\cite{GJ_2009,JG_2010})  over translation invariant quantum networks on a common carrier lattice in terms of algebraic operations over their spatio-temporal transfer functions and energy parameters. Network interconnections  arise in quantum control settings, where performance specifications include stability and minimization of cost functionals\cite{ZJ_2012}.

Under a stability condition in the spatial frequency domain, the network  has an invariant multipoint Gaussian quantum state\cite{P_2010,VPJ_2018a} in the case of statistically independent vacuum input fields. We consider a quadratic function of the network variables of interest with a block Toeplitz weighting matrix for a finite fragment of the lattice and over a bounded time interval. The tail probabilities for this self-adjoint quantum variable admit upper bounds involving its exponential moments, which, similarly to Ref. \refcite{VPJ_2018a}, lead to a quadratic exponential functional (QEF) as a risk-sensitive performance criterion for finite fragments of the network over finite time horizons. 

The QEF is a quantum mechanical counterpart of the cost functionals used in classical risk-sensitive control\cite{BV_1985,J_1973,W_1981} which has links with minimax linear-quadratic-Gaussian control\cite{DJP_2000,P_2006,PJD_2000} addressing the issue of system robustness against  statistical uncertainties with a relative entropy description. The latter has its analogue in terms of quantum relative entropy\cite{OW_2010} leading to similar robustness properties\cite{VPJ_2018b} in the context of risk-sensitive quantum feedback control and filtering problems\cite{B_1996,J_2004,J_2005,YB_2009}, some of which employ a different yet related\cite{VPJ_2019a} class of time-ordered exponentials.

Assuming the invariant state of the network, we study the spatio-temporal asymptotic rate of the QEF per unit time and per lattice site in the thermodynamic limit\cite{R_1978} of unboundedly growing time horizons and fragments of the network. The resulting spatio-temporal frequency domain formula for the QEF rate  is organised as an integral of the log-determinant  of a matrix-valued function over the product of the multidimensional torus with the frequency axis. The integrand     involves two spectral functions, which are  associated with the real and imaginary parts of the invariant quantum covariance kernel of the network variables and form their quantum spectral density. One of these matrix-valued spectral functions, originating from the two-point commutator kernel,  enters the frequency-domain representation of the QEF rate in composition with trigonometric functions\cite{H_2008}. Combined with the multivariate nature of the integral,  this makes the evaluation of the QEF rate inaccessible to the standard application of the residue theorem. We obtain a differential equation and an asymptotic expansion for the QEF rate as a function of the risk sensitivity parameter, which can be used for its numerical computation, similar to the homotopy methods for solving parameter dependent algebraic equations\cite{MB_1985}.

Continuing the development of methods for computing the QEFs,    this paper employs a number of results from  a series of recent publications on Lie-algebraic techniques\cite{VPJ_2019a}, parametric randomization\cite{VPJ_2018c} and quantum Karhunen-Loeve  expansions\cite{VPJ_2019b,VJP_2019} developed for this purpose. These results have led to an integral operator representation of the QEF\cite{VPJ_2019c} and a frequency-domain formula\cite{VPJ_2020b} for their infinite time horizon rates for OQHOs with finitely many degrees of freedom in Gaussian quantum states, which has been extended to more general Gaussian quantum processes in Ref. \refcite{VPJ_2021}. In addition to their relevance to quantum risk-sensitive control,  these approaches have deep connections with operator exponential structures studied in mathematical physics and quantum probability (for example, in the context of operator algebras\cite{AB_2018}, moment-generating functions for quadratic Hamiltonians\cite{PS_2015} and the quantum L\'{e}vy area\cite{CH_2013,H_2018}).

The paper is organised as follows.
Section~\ref{sec:net} specifies the class of translation invariant quantum networks being considered.
Section~\ref{sec:PR} represents PR conditions for the network in the spatial frequency domain.
Section~\ref{sec:par} provides a parameterization of the network QSDEs in terms of the energy and coupling matrices and outlines their computation for interconnections of networks.
Section~\ref{sec:inv} considers the invariant Gaussian state of the network, satisfying a stability condition and  driven by vacuum fields.
Section~\ref{sec:QEF} specifies QEFs for finite fragments of the network over bounded time intervals and clarifies their role for large deviations estimates for network trajectories.
Sections~\ref{sec:temp} and \ref{sec:spat}  establish the temporal and spatio-temporal QEF growth rates.
Section~\ref{sec:homo}  discusses the computation of the QEF rate using homotopy and asymptotic expansion techniques.
Section~\ref{sec:conc} makes concluding remarks.
\ref{sec:Toeplitz} to \ref{sec:aver}
provide
subsidiary material (on block Toeplitz matrices, and averaging for trace-analytic functionals of such matrices and integral operators) and some of the particularly long proofs.

\section{Translation Invariant Quantum Network}
\label{sec:net}

We consider a network of identical linear quantum stochastic systems  at sites of a $\nu$-dimensional integer lattice $\mZ^\nu$. For any $j\in \mZ^\nu$, the $j$th component system is a multi-mode open quantum harmonic oscillator (OQHO) with an even number $n$ of internal dynamic variables which are time-varying self-adjoint operators on (a dense domain of) a Hilbert space $\fH$. These system variables are assembled into a vector\footnote{vectors are organised as columns unless specified otherwise} $X_j(t)$ (the time argument $t\> 0$ will often be omitted for brevity) and act initially (at $t=0$) on a copy $\fH_j$ of a common complex separable Hilbert space.   It is assumed that they satisfy the canonical commutation relations (CCRs)
\begin{equation}
\label{Xcomm}
    [X_j(t), X_k(t)^\rT] = 2i\delta_{jk}\Theta,
    \qquad
    j,k\in \mZ^\nu,
    \quad
    t\>0,
\end{equation}
where the transpose  $(\cdot)^\rT$
applies to matrices and vectors of operators as if the latter were scalars,
$i:=\sqrt{-1}$ is the imaginary unit, $\delta_{jk}$ is the Kronecker delta,
and $\Theta$ is a nonsingular real antisymmetric matrix of order $n$. Here,
$[\alpha, \beta^\rT] := ([\alpha_a, \beta_b])_{1\< a\< r, 1\< b\< s} = \alpha\beta^\rT - (\beta\alpha^\rT)^\rT $ is the matrix of commutators $[\alpha_a, \beta_b] = \alpha_a \beta_b-\beta_b\alpha_a$ for vectors $\alpha:= (\alpha_a)_{1\< a\< r}$,  $\beta:= (\beta_b)_{1\< b\< s}$ formed from linear operators.

In particular, if the internal variables of the component system are the quantum mechanical positions and  momenta\cite{S_1994}  $q_1, \ldots, q_{n/2}$ and $p_1:= -i\d_{q_1}, \ldots, p_{n/2}:= -i\d_{q_{n/2}}$ on the  Schwartz space\cite{V_2002}, then the CCR matrix takes the form $\Theta = \frac{1}{2}\bJ\ox I_{n/2}$, where $\ox$ is the Kronecker product, the matrix
\begin{equation}
\label{bJ}
    \bJ :=
{\begin{bmatrix}
  0 & 1\\
  -1 & 0
\end{bmatrix} }
\end{equation}
spans the subspace of antisymmetric matrices of order $2$,   and $I_r$ is the identity matrix of order $r$. However, this special structure of $\Theta$ is not assumed in the general case considered in what follows.

In addition to the internal variables, the $j$th OQHO has multichannel  input and output bosonic fields $W_j$, $Y_j$ which consist of $m$  and $r$ time-varying  self-adjoint quantum variables, respectively (the dimensions  $m$, $r$ are even and satisfy $r\< m$).  The input field $W_j$ is a quantum Wiener process on a symmetric Fock space $\fF_j$.  The network-field space   has the tensor-product structure $\fH := \ox_{j\in \mZ^\nu}(\fH_j \ox \fF_j)$, with the composite Fock space $\fF:= \ox_{j\in \mZ^\nu} \fF_j$ accommodating the input fields. These fields satisfy the two-point CCRs
\begin{equation}
\label{Wcomm}
  [W_j(s), W_k(t)^\rT]
  =
  2i\delta_{jk}
  \min(s,t)
  J_m,
      \qquad
    j,k\in \mZ^\nu,
    \
    s,t\>0,
\end{equation}
where
\begin{equation}
\label{Jm}
    J_m := \bJ \ox I_{m/2}
    =
    \begin{bmatrix}
      0 & I_{m/2}\\
      - I_{m/2} & 0
    \end{bmatrix}
\end{equation}
is an orthogonal  real antisymmetric matrix of order $m$ defined in terms of (\ref{bJ}), so that $J_m^2 = -I_m$. The right-hand side of (\ref{Wcomm}) vanishes at $s=0$ or $t=0$ since the initial input field operators act as the identity operator on $\fF$, which commutes with any operator. Due to the continuous tensor-product     structure\cite{PS_1972} of the Fock space filtration,  the relation (\ref{Wcomm}) is equivalent to its fulfillment for all $s=t\>0$, whose incremental form is given by
\begin{align}
\nonumber
    \rd [W_j, W_k^\rT]
    & =
    [\rd W_j, W_k^\rT]
    +
    [W_j, \rd W_k^\rT]
    +
  [\rd W_j, \rd W_k^\rT]\\
\label{dWcomm}
  & =
    [\rd W_j, \rd W_k^\rT]
    =
  2i\delta_{jk}
  J_m
  \rd t.
\end{align}
Here, use is also made of the quantum Ito lemma\cite{HP_1984} and the property of
the future-pointing Ito increments of the input quantum Wiener processes to commute with adapted processes (in the sense of the filtration of the network-field space $\fH$).
In particular,
\begin{equation}
\label{WXYdWcomm}
    [W_j(s), \rd W_k(t)^{\rT}] = 0,
    \qquad
    [X_j(s), \rd W_k(t)^{\rT}] = 0,
    \qquad
    [Y_j(s), \rd W_k(t)^{\rT}] = 0
\end{equation}
for all     $j,k\in \mZ^\nu$,     $t\> s\> 0$.
We model the Heisenberg evolution of the network by a denumerable
set of linear quantum stochastic  differential equations (QSDEs)
\begin{align}
\label{dXj}
  \rd X_j
  & =
  \sum_{k \in \mZ^\nu} (A_{j-k}X_k \rd t + B_{j-k} \rd W_k), \\
\label{dYj}
  \rd Y_j
  & =
  \sum_{k \in \mZ^\nu} (C_{j-k}X_k \rd t + D_{j-k} \rd W_k),
  \qquad
  j \in \mZ^\nu,
\end{align}
which are coupled to each other and driven by the input fields in a translation invariant fashion. Their coefficients
are specified by the matrices
\begin{equation}
\label{ABCD}
    A_\ell \in \mR^{n\x n},
    \qquad
    B_\ell\in \mR^{n\x m},
    \qquad
    C_\ell\in \mR^{r\x n},
    \qquad
    D_\ell \in \mR^{r\x m},
\end{equation}
which depend on the relative location $\ell \in \mZ^\nu$ of the lattice sites. For what follows, these  matrices are assumed to be absolutely summable over $\ell\in\mZ^\nu$, which is equivalent to
\begin{equation}
\label{ABCDsum}
    \sum_{\ell \in \mZ^\nu}
    \left\|
    {\begin{bmatrix}
      A_\ell & B_\ell\\
      C_\ell & D_\ell
    \end{bmatrix}}
    \right\|
    < +\infty,
\end{equation}
where $\|\cdot \|$ is the operator norm   (the largest singular value) of a matrix. The particular choice of a matrix norm does not affect the validity of (\ref{ABCDsum}).

The set of QSDEs (\ref{dXj}), (\ref{dYj}) can be represented formally in terms of the augmented vectors $X:= (X_k)_{k \in \mZ^\nu}$, $W:= (W_k)_{k \in \mZ^\nu}$, $Y:= (Y_k)_{k \in \mZ^\nu}$ of the internal variables and external fields of the network as
\begin{align}
\label{dX}
  \rd X
  & =
  AX\rd t + B \rd W, \\
\label{dY}
  \rd Y
  & =
  CX \rd t + D\rd W,
\end{align}
where $A:= (A_{j-k})_{j,k\in \mZ^\nu} \in \fT_{n,n}$, $B:= (B_{j-k})_{j,k\in \mZ^\nu} \in \fT_{n,m}$, $C:= (C_{j-k})_{j,k\in \mZ^\nu}\in \fT_{r,n}$, $D:= (D_{j-k})_{j,k\in \mZ^\nu}\in \fT_{r,m}$ are appropriately dimensioned real block Toeplitz matrices with finite norms $\|A\|_1$, $\|B\|_1$, $\|C\|_1$, $\|D\|_1$ in view of (\ref{ABCDsum}), (\ref{f1}); see \ref{sec:Toeplitz}.
The absolute summability condition secures well-posedness of
the spatial Fourier transforms (SFTs)
\begin{equation}
\label{cABCD}
    {\begin{bmatrix}
      \cA(\sigma) & \cB(\sigma)\\
      \cC(\sigma) & \cD(\sigma)
    \end{bmatrix}}
    :=
    \sum_{\ell\in \mZ^\nu}
    \re^{-i\ell^\rT \sigma}
        {\begin{bmatrix}
      A_\ell & B_\ell\\
      C_\ell & D_\ell
    \end{bmatrix}},
    \qquad
    \sigma \in \mT^\nu,
\end{equation}
so that
$\cA$, $\cB$, $\cC$, $\cD$ are appropriately dimensioned complex matrix-valued functions, continuous and $2\pi$-periodic over their $\nu$ variables.
The  matrices in (\ref{ABCD}) are recovered from (\ref{cABCD}) through the inverse SFT as
\begin{equation*}
\label{inv}
{\begin{bmatrix}
      A_\ell & B_\ell\\
      C_\ell & D_\ell
    \end{bmatrix}}
    =
    \frac{1}{(2\pi)^\nu}
    \int_{\mT^\nu}
    \re^{i\ell^\rT \sigma}
    {\begin{bmatrix}
      \cA(\sigma) & \cB(\sigma)\\
      \cC(\sigma) & \cD(\sigma)
    \end{bmatrix}}
    \rd \sigma.
\end{equation*}
Since the matrices (\ref{ABCD}) are real,
their SFTs $\cA$, $\cB$, $\cC$, $\cD$ are Hermitian in the sense that  $\overline{\cA(\sigma)} = \cA(-\sigma)$ for all $\sigma \in \mT^\nu$ (and similarly for $\cB$, $\cC$, $\cD$), and hence,
\begin{align}
\label{cAB*}
    \cA(\sigma)^*
    & = \cA(-\sigma)^\rT,
    \qquad\!\!\!
    \cB(\sigma)^* = \cB(-\sigma)^\rT,    \\
\label{cCD*}
    \cC(\sigma)^*
    & = \cC(-\sigma)^\rT,
    \qquad
    \cD(\sigma)^* = \cD(-\sigma)^\rT
\end{align}
for all $\sigma\in\mT^\nu$. The right-hand sides of (\ref{cAB*}), (\ref{cCD*}) are the SFTs of the matrices $A_{-\ell}^\rT$,  $B_{-\ell}^\rT$, $C_{-\ell}^\rT$, $D_{-\ell}^\rT$ which constitute $A^\rT$, $B^\rT$, $C^\rT$, $D^\rT$,  respectively. Dynamic properties of the translation invariant network can be represented in the spatial frequency domain using the SFTs $\cA$, $\cB$, $\cC$, $\cD$. Such properties include the preservation of commutation relations.

\section{Physical Realizability Conditions in the Spatial Frequency Domain}
\label{sec:PR}

Similarly to OQHOs with a finite number of external field channels and internal dynamic variables,  the matrices (\ref{ABCD}) of the network QSDEs (\ref{dXj}), (\ref{dYj})  satisfy physical realizability (PR) conditions which reflect the preservation of the CCRs (\ref{Xcomm}) together with
\begin{equation}
\label{XYcomm}
    [X_j(t), Y_k(s)^{\rT}] = 0,
    \qquad
    j,k\in \mZ^\nu,\
    t\> s\> 0.
\end{equation}
The fulfillment of (\ref{XYcomm}) at $s=t=0$ is secured by the commutativity of operators, acting on different initial and Fock spaces $\fH_j$, $\fF_k$ and appropriately extended to $\fH_j\ox \fF_k$. An additional PR condition\footnote{which is important in the context of concatenating quantum networks as input-output maps, considered in Section~\ref{sec:par}} comes from the requirement that the commutation structure of the output fields of the network is similar to that of the input fields in (\ref{Wcomm}), (\ref{dWcomm}):
\begin{equation}
\label{Ycomm}
  [Y_j(s), Y_k(t)^\rT]
  =
  2i\delta_{jk}
  \min(s,t)
  J_r,
      \qquad
    j,k\in \mZ^\nu,
    \
    s,t\>0,
\end{equation}
where $J_r$ is defined according to (\ref{Jm}). The following theorem represents the PR conditions in the spatial frequency domain as a network counterpart of the previous results for OQHOs with a finite number of variables\cite{JNP_2008,SP_2012} and extends Ref. \refcite{VP_2014}.

\begin{theorem}
\label{th:PR}
The network QSDEs (\ref{dXj}), (\ref{dYj}) preserve the CCRs (\ref{Xcomm}), (\ref{XYcomm}), (\ref{Ycomm}) if and only if the SFTs (\ref{cABCD}) satisfy
\begin{align}
\label{PR1}
  \cA(\sigma) \Theta + \Theta \cA(\sigma)^* + \cB(\sigma) J_m \cB(\sigma)^*
  & = 0,\\
\label{PR2}
  \Theta \cC(\sigma)^*  + \cB(\sigma) J_m \cD(\sigma)^*
  & = 0,\\
\label{PR3}
  \cD(\sigma) J_m \cD(\sigma)^*
  & = J_r
\end{align}
for all $\sigma \in \mT^\nu$. \hfill$\square$
\end{theorem}

As can be seen from the proof of this theorem in \ref{sec:PRproof}, the PR conditions (\ref{PR1})--(\ref{PR3}) are obtained by applying the homomorphism between the algebra of block Toeplitz matrices, the corresponding  convolution algebra of matrix-valued maps on the lattice $\mZ^\nu$ and the algebra of SFTs with the pointwise multiplication over the torus $\mT^\nu$ to the PR conditions
\begin{align}
\label{bPR1}
    A \bTheta + \bTheta A^\rT + B \bitJ_m B^\rT
   & = 0,\\
\label{bPR2}
  \bTheta C^\rT  + B\bitJ_m D^\rT
   & = 0,\\
\label{bPR3}
  D\bitJ_m D^\rT
   & = \bitJ_r
\end{align}
for the QSDEs  (\ref{dX}), (\ref{dY}). Here, the block diagonal matrices
\begin{equation}
\label{TJJ}
    \bTheta
    :=
    (\delta_{jk}\Theta)_{j,k\in \mZ^\nu},
     \qquad
     \bitJ_m:= (\delta_{jk}J_m)_{j,k\in \mZ^\nu},
      \qquad
      \bitJ_r:= (\delta_{jk}J_r)_{j,k\in \mZ^\nu}
\end{equation}
specify the CCRs for the internal network variables and the external fields, respectively:
\begin{equation*}
\label{XXcomm}
    [X,X^\rT] = 2i\bTheta,
    \qquad
    [\rd W,\rd W^\rT] = 2i\bitJ_m\rd t,
    \qquad
    [\rd Y,\rd Y^\rT] = 2i\bitJ_r\rd t.
\end{equation*}
Indeed, the matrices $\bTheta$, $\bitJ_m$, $\bitJ_r$ in (\ref{TJJ}) have constant SFTs $\Theta$, $J_m$, $J_r$, respectively, which together with (\ref{cAB*}), (\ref{cCD*}),  makes (\ref{PR1})--(\ref{PR3}) equivalent to the corresponding conditions in (\ref{bPR1})--(\ref{bPR3}).
Also note that the PR conditions (\ref{PR1})--(\ref{PR3}) in the spatial frequency domain can be represented in the form
\begin{equation}
\label{PRmat}
    \begin{bmatrix}
      \cA(\sigma) & \cB(\sigma) & I_n & 0\\
      \cC(\sigma) & \cD(\sigma) & 0 & I_r
    \end{bmatrix}
    \begin{bmatrix}
      0 & 0& \Theta  & 0\\
      0 & J_m & 0 & 0 \\
      \Theta  & 0 & 0 & 0\\
       0& 0 & 0 & -J_r
    \end{bmatrix}
    \begin{bmatrix}
      \cA(\sigma)^* & \cC(\sigma)^*\\
      \cB(\sigma)^* & \cD(\sigma)^*\\
      I_n & 0\\
      0 & I_r
    \end{bmatrix}
    =
    0,
    \qquad
    \sigma \in \mT^\nu.
\end{equation}
Similarly to OQHOs with finitely many dynamic variables\cite{SP_2012}, the PR conditions (\ref{PR1})--(\ref{PR3}) (or (\ref{PRmat})) imply a $(J,J)$-unitarity property\cite{K_1997}  for the spatio-temporal transfer function of the network from $\rd W$ in (\ref{dX}) to $\rd Y$ in (\ref{dY}) defined as
\begin{equation}
\label{trans}
    F(\sigma,s)
     = \cC(\sigma)(sI_n- \cA(\sigma))^{-1} \cB(\sigma) + \cD(\sigma),
     \qquad
     \sigma \in \mT^\nu,\
     s \in \mC\setminus \fS(\sigma),
\end{equation}
by analogy with the finite-dimensional case, where $\fS(\sigma)$ denotes the spectrum of the matrix $\cA(\sigma)$. The corresponding conjugate of the transfer function  is given by
\begin{align}
\nonumber
  F^\diam(\sigma,s)
  & :=
  F(\sigma,-\overline{s})^*\\
\nonumber
  & =
  -\cB(\sigma)^*(sI_n+ \cA(\sigma)^*)^{-1} \cC(\sigma)^* + \cD(\sigma)^*\\
\label{Fdiam}
  & =
  F(-\sigma,-s)^\rT
\end{align}
for any $s \in \mC\setminus (-\fS(-\sigma)) $ in view of the relations (\ref{cAB*}), (\ref{cCD*}) and the invariance of the spectrum of a square matrix under the transpose.

\begin{theorem}
\label{th:JJ}
Under the PR conditions (\ref{PR1})--(\ref{PR3}) on the network QSDEs (\ref{dXj}), (\ref{dYj}), the transfer function (\ref{trans}) satisfies
\begin{equation}
\label{JJ}
  F(\sigma,s)J_m F^\diam(\sigma,s) = J_r,
  \qquad
    \sigma \in \mT^\nu,\
  s \in \mC\setminus (\fS(\sigma) \bigcup (-\fS(-\sigma))).
\end{equation}
\hfill$\square$
\end{theorem}
\begin{proof}
We will use an auxiliary spatio-temporal transfer function $\cT$ (from $B \rd W$ in (\ref{dX}) to the drift $CX$ of $\rd Y$ in (\ref{dY}))  and its conjugate $\cT^\diam$ given by
\begin{equation}
\label{cZ}
    \cT(\sigma,s)
     :=
    \cC(\sigma)(sI_n - \cA(\sigma))^{-1},
    \qquad
    \cT^\diam(\sigma,s)
     :=
    -(sI_n + \cA(\sigma)^*)^{-1} \cC(\sigma)^*.
\end{equation}
A combination of (\ref{trans}), (\ref{cZ}) leads to the identity
\begin{align}
\label{sZ_ZF}
    \begin{bmatrix}
        \cT(\sigma,s)  & I_r
    \end{bmatrix}
    \begin{bmatrix}
      \cA(\sigma) & \cB(\sigma)\\
      \cC(\sigma) & \cD(\sigma)
    \end{bmatrix}
    & =
    \begin{bmatrix}
      s\cT(\sigma,s) & \cF(\sigma,s)
    \end{bmatrix}
\end{align}
and its conjugate counterpart
\begin{align}
\label{sZ1_ZF1}
    \begin{bmatrix}
      \cA(\sigma)^* &      \cC(\sigma)^*\\
      \cB(\sigma)^* &      \cD(\sigma)^*
    \end{bmatrix}
    \begin{bmatrix}
        \cT^\diam(\sigma,s)  \\
         I_r
    \end{bmatrix}
    & =
    \begin{bmatrix}
      -s\cT^\diam(\sigma,s)\\
      \cF^\diam(\sigma,s)
    \end{bmatrix}.
\end{align}
Since the fulfillment of (\ref{PR1})--(\ref{PR3}) is equivalent to (\ref{PRmat}), then by left and right multiplying both sides of (\ref{PRmat})
by
$
    {\begin{bmatrix}
        \cT(\sigma,s)  & I_r
    \end{bmatrix}}
$,
$
    {\begin{bmatrix}
        \cT^\diam(\sigma,s)\\
        I_r
    \end{bmatrix}}
$
and using (\ref{sZ_ZF}), (\ref{sZ1_ZF1}), it follows that
\begin{align}
\nonumber
    0 & =
        \begin{bmatrix}
      s\cT(\sigma,s) & \cF(\sigma,s) & \cT(\sigma,s) & I_r
    \end{bmatrix}
    \begin{bmatrix}
      0 & 0& \Theta  & 0\\
      0 & J_m & 0 & 0 \\
      \Theta  & 0 & 0 & 0\\
       0& 0 & 0 & -J_r
    \end{bmatrix}
    \begin{bmatrix}
      -s\cT^\diam(\sigma,s) \\
       \cF^\diam(\sigma,s) \\
        \cT^\diam(\sigma,s) \\
         I_r
    \end{bmatrix}\\
\nonumber
    & =
    \cF(\sigma,s)J_m\cF^\diam(\sigma,s) - J_r
\end{align}
for $(\sigma,s)$ belonging to the intersection of domains of the functions $\cF$,  $\cF^\diam$ in (\ref{trans}), (\ref{Fdiam}),
which establishes (\ref{JJ}).
\end{proof}

The validity of (\ref{JJ}), as a corollary of the PR conditions, does not employ a particular form of the CCR matrix $\Theta$ of the internal variables and is a property   of the network as an input-output operator. Also note that (\ref{PRmat}), (\ref{JJ}) are organised as indefinite quadratic constraints on the quadruple $(\cA, \cB, \cC, \cD)$ of the SFTs and the transfer function $\cF$.

\section{Energy and Coupling Matrices, and Network Interconnections}
\label{sec:par}

The fulfillment of the PR conditions (\ref{PR1}), (\ref{PR2}) is secured by the parameterisation of the coefficients (\ref{ABCD}) of the QSDEs (\ref{dXj}), (\ref{dYj}) in terms of energy and coupling matrices $R:=(R_{j-k})_{j,k\in \mZ^\nu} = R^\rT \in \fT_{n,n}$ and $M:=(M_{j-k})_{j,k\in \mZ^\nu} \in \fT_{m,n}$ specifying the network Hamiltonian and  the operators of coupling of the component  systems to the input fields. More precisely, in accordance with (\ref{PR1}), (\ref{PR2}),
\begin{align}
\label{cA}
    \cA(\sigma)
    & = 2\Theta (\cR(\sigma) + \cM(\sigma)^* J_m \cM(\sigma)),\\
\label{cB}
    \cB(\sigma)
    & = 2\Theta \cM(\sigma)^* ,\\
\label{cC}
    \cC(\sigma)
    & = 2\cD(\sigma)J_m \cM(\sigma) ,
    \qquad
    \sigma \in \mT^\nu,
\end{align}
where $\cR$, $\cM$ are the SFTs associated with  $R$, $M$, respectively. The blocks $R_\ell = R_{-\ell}^\rT \in \mR^{n\x n}$ of the energy matrix  $R$
parameterise the Hamiltonian
\begin{align}
\nonumber
    H_G
    & :=
    \frac{1}{2}
    X_G^\rT
    R_G
    X_G\\
\nonumber
    & = \frac{1}{2}
    \sum_{j,k\in G}
    X_j^\rT R_{j-k} X_k\\
\label{HG}
    & =
    \frac{1}{2}
    \sum_{j\in G}
    X_j^\rT R_0 X_j
    +
    \frac{1}{2}
    \sum_{j,k \in G, j\ne k}
    X_j^\rT R_{j-k} X_k
\end{align}
for the fragment of the network on a nonempty finite subset $G \subset \mZ^\nu$ of the lattice consisting of $\#G <+\infty$ sites, where the relevant network variables are assembled into the vector
\begin{equation}
\label{XG}
  X_G : = (X_k)_{k \in G},
\end{equation}
and use is made of the matrix $R_G:= (R_{j-k})_{j,k\in G} = R_G^\rT \in \mR^{G\x G}$.
In the Hamiltonian (\ref{HG}), the matrix  $R_0$ specifies the self-energy of the component systems, while $R_{j-k}$ parameterises the direct (energy) coupling of the $j$th and $k$th systems, with  $j\ne k$. For any $j,k \in \mZ^\nu$, the matrix $M_{j-k}$ specifies the vector $M_{j-k}X_k$ of operators of coupling of the $k$th component system to the input field $W_j$. Therefore, (\ref{cA})--(\ref{cC}) are equivalent to
\begin{align}
\label{Aell}
    A_\ell
    & =
    2\Theta
    \Big(
        R_\ell+
        \sum_{c\in \mZ^\nu}
        M_c^\rT
        J_m
        M_{\ell+c}
    \Big),\\
\label{Bell}
    B_\ell
    & = 2\Theta M_{-\ell}^\rT,\\
\label{Cell}
    C_\ell
    & =
    2
    \sum_{c\in \mZ^\nu}
    D_{\ell-c}J_m M_c,
    \qquad
    \ell \in \mZ^\nu.
\end{align}
In the case of \emph{finite range interaction} (between the component systems in the network and with the external fields), the matrices $R_\ell$, $M_\ell$, $D_\ell$  vanish for all sufficiently large $\ell \in \mZ^\nu$, and hence, so also do the matrices $A_\ell$, $B_\ell$, $C_\ell$ in (\ref{Aell})--(\ref{Cell}). In particular, a network with nearest neighbour coupling between the subsystems, with each of them being affected by the local input field, is illustrated in Fig.~\ref{fig:net}.
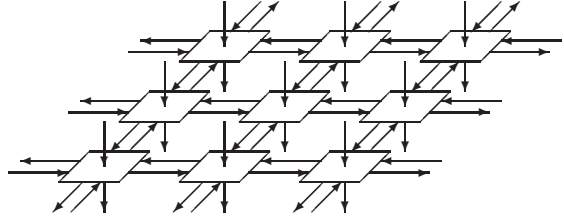
\begin{figure}[h!]
\centering
\unitlength=0.8mm
\linethickness{0.4pt}
\begin{picture}(100.00,31.00)

\multiput(15,2)(10,10){3}{
\multiput(0,0)(20,0){3}
{
\put(0, 0){\line(1,0){10}}
\put(10, 0){\line(1,1){5}}
\put(5, 5){\line(1,0){10}}
\put(0, 0){\line(1,1){5}}

\put(7.5, 10){\vector(0,-1){7.5}}
\put(7.5, 0){\vector(0,-1){5}}

\put(-8.5, 1.5){\vector(1,0){10}}
\put(3.5, 3.5){\vector(-1,0){10}}

\put(11.5, 5){\vector(1,1){5}}
\put(13.5, 10){\vector(-1,-1){5}}
}
\put(51.5, 1.5){\vector(1,0){10}}
\put(63.5, 3.5){\vector(-1,0){10}}
}

\multiput(15,2)(20,0){3}{
\put(2, -5){\vector(1,1){5}}
\put(4, 0){\vector(-1,-1){5}}}
\end{picture}\vskip3mm
\caption{An illustration of a $(3\x 3)$-fragment of a two-dimensional ($\nu = 2$) quantum network, where each component system is coupled to its nearest neighbours and a local input field (the external quantum fields are represented by vertical arrows).}
\label{fig:net}
\end{figure}

We will now outline the computation of energy and coupling matrices for network interconnections. Consider two translation invariant quantum networks on the common lattice $\mZ^\nu$ with  quadruples $(A^{[k]}, B^{[k]}, C^{[k]}, D^{[k]})$ of block Toeplitz matrices and input, internal and output dimensions  $m_k$, $n_k$, $r_k$, respectively, $k=1,2$.    The corresponding augmented vectors of input, internal and output fields are denoted by $W^{[k]}$, $X^{[k]}$, $Y^{[k]}$, and the CCR matrices of the internal variables of the networks in the sense of (\ref{Xcomm}) are denoted by $\Theta_k$.  The spatio-temporal transfer functions of the networks are
\begin{equation}
\label{Fk}
    F_k(\sigma,s):= \cC_k(\sigma)(sI_{n_k}- \cA_k(\sigma))^{-1} \cB_k(\sigma) + \cD_k(\sigma)
\end{equation}
with values in $\mC^{r_k\x m_k}$ for $s\in \mC$ with $\Re s$ sufficiently large, and $\sigma \in \mT^\nu$. If $r_1=m_2$, and the output fields of the first network are fed as the input fields to the second network (see Fig.~\ref{fig:FF}),
\begin{figure}[htbp]
\unitlength=1mm
\linethickness{0.4pt}
\centering
\begin{picture}(110,11.00)
    \put(35,0){\framebox(10,10)[cc]{{$F_2$}}}
    \put(65,0){\framebox(10,10)[cc]{{$F_1$}}}
    \put(35,5){\vector(-1,0){20}}
    \put(65,5){\vector(-1,0){20}}
    \put(95,5){\vector(-1,0){20}}
    \put(10,5){\makebox(0,0)[cc]{$Y^{[2]}$}}
    \put(100,5){\makebox(0,0)[cc]{$W^{[1]}$}}
    \put(55,10){\makebox(0,0)[cc]{$W^{[2]}=Y^{[1]}$}}
\end{picture}
\caption{The concatenation of translation invariant quantum networks on a common lattice with spatio-temporal transfer functions $F_1$, $F_2$.}
\label{fig:FF}
\end{figure}
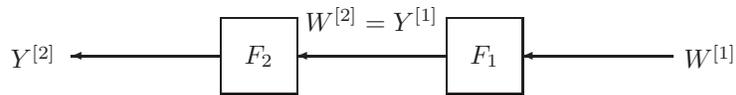
the resulting composition is a translation invariant quantum network with input, internal and output dimensions $m_1$, $n:= n_1+n_2$, $r_2$, respectively, and the spatio-temporal transfer function
$$
    F(\sigma,s)
    =
    F_2(\sigma,s)F_1(\sigma,s)
    =
    \cC(\sigma)(sI_n- \cA(\sigma))^{-1} \cB(\sigma) + \cD(\sigma),
$$
which is the pointwise product of the transfer functions (\ref{Fk}). Here, as in the case of cascaded classical linear time invariant systems,
\begin{equation}
\label{ser}
    {\begin{bmatrix}
      \cA(\sigma) & \cB(\sigma)\\
      \cC(\sigma) & \cD(\sigma)
    \end{bmatrix}}
    =
    \left[
    {\begin{array}{cc|c}
      \cA_1(\sigma) & 0& \cB_1(\sigma)\\
      \cB_2(\sigma)\cC_1(\sigma) & \cA_2(\sigma) & \cB_2(\sigma)\cD_1(\sigma) \\
      \hline
      \cD_2(\sigma)\cC_1(\sigma) & \cC_2(\sigma) & \cD_2(\sigma)\cD_1(\sigma)
    \end{array}}
    \right],
    \qquad
    \sigma \in\mT^\nu.
\end{equation}
The concatenated network  has the CCR matrix
$$
    \Theta
    =
    \begin{bmatrix}
      \Theta_1 & 0\\
      0 & \Theta_2
    \end{bmatrix}
$$
for its internal variables in the sense of (\ref{Xcomm})
and the energy and coupling matrices which, in view of (\ref{cA})--(\ref{cC}) and (\ref{ser}),  can be recovered from the SFTs
\begin{align*}
%\label{cR}
    \cR(\sigma)
    & =
    \begin{bmatrix}
      \cR_1(\sigma) & -\cM_1(\sigma)^* J_{m_1} \cD_1(\sigma)^* \cM_2(\sigma)\\
      \cM_2(\sigma)^*\cD_1(\sigma) J_{m_1} \cM_1(\sigma) & \cR_2(\sigma)
    \end{bmatrix},\\
%\label{cM}
    \cM(\sigma)
    & =
    \begin{bmatrix}
      \cM_1(\sigma) & \cD_1(\sigma)^* \cM_2(\sigma)
    \end{bmatrix},
    \qquad
    \sigma \in \mT^\nu,
\end{align*}
which are expressed in terms of the SFTs $\cR_k$, $\cM_k$ of the energy and coupling matrices of the networks,  $k=1,2$.
Other algebraic operations for translation invariant networks on $\mZ^\nu$ are carried out in a similar pointwise fashion over the torus $\mT^\nu$. For example, feedback interconnections of   such networks involve linear fractional transformations of spatio-temporal transfer functions. In particular, Fig.~\ref{fig:loop}
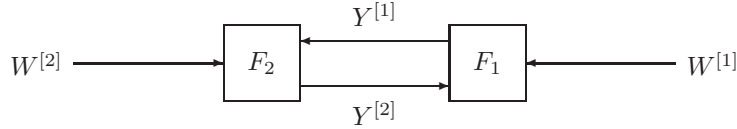
\begin{figure}[htbp]
\unitlength=1mm
\linethickness{0.4pt}
\centering
\begin{picture}(110,11.00)
    \put(35,0){\framebox(10,10)[cc]{{$F_2$}}}
    \put(65,0){\framebox(10,10)[cc]{{$F_1$}}}
    \put(15,5){\vector(1,0){20}}
    \put(65,8){\vector(-1,0){20}}
    \put(45,2){\vector(1,0){20}}
    \put(95,5){\vector(-1,0){20}}
    \put(10,5){\makebox(0,0)[cc]{$W^{[2]}$}}
    \put(100,5){\makebox(0,0)[cc]{$W^{[1]}$}}
    \put(55,10){\makebox(0,0)[cb]{$Y^{[1]}$}}
    \put(55,0){\makebox(0,0)[ct]{$Y^{[2]}$}}
\end{picture}\vskip2mm
\caption{A field-mediated feedback interconnection of translation invariant quantum networks on a common lattice with external quantum fields $W^{[1]}$, $W^{[2]}$.}
\label{fig:loop}
\end{figure}
illustrates a quantum feedback network, resulting from a field-mediated connection of a translation invariant network $F_1$, interpreted as a plant,  with another such network $F_2$ (on the same carrier lattice $\mZ^\nu$), playing the role of a controller. This gives rise to coherent (measurement-free) quantum control settings\cite{NJP_2009,SVP_2015}, where the energy parameters  of the controller and its coupling with the plant can be varied so as to satisfy performance specifications for the closed-loop network such as stability and minimization of cost functionals in the steady-state regime.

\section{Invariant Gaussian State in the Case of Vacuum Input Fields}
\label{sec:inv}

We will be concerned with the case of statistically independent input fields in the vacuum state, defined in terms of the quasi-characteristic functional (QCF)\cite{CH_1971,HP_1984,P_1992} of the incremented quantum Wiener processes as
\begin{align}
\nonumber
    \bE
    \re^{i\int_0^T u(t)^\rT \rd W(t)}
    & =
    \bE
    \prod_{k \in \mZ^\nu}
    \re^{i\int_0^T u_k(t)^\rT \rd W_k(t)}    \\
\label{vac}
    & =
    \prod_{k \in \mZ^\nu}
    \bE
    \re^{i\int_0^T u_k(t)^\rT \rd W_k(t)}
    =
    \re^{-\frac{1}{2} \int_0^T |u(t)|^2\rd t}
\end{align}
for any time horizon $T>0$ and any square integrable map $u:=(u_k)_{k \in \mZ^\nu}: [0,T]\to \ell^2(\mZ^\nu, \mR^m)$, where the standard Euclidean norm $|\cdot|$ is extended to $\ell^2(\mZ^\nu, \mR^m)$ as $|u|:= \sqrt{\sum_{k \in \mZ^\nu}|u_k|^2}$ along with the inner product $u^\rT w$.  Here, $\bE \zeta := \Tr(\rho \zeta)$ is the quantum expectation over the density operator
\begin{equation}
\label{rho}
    \rho:= \rho_0\ox \ups,
\end{equation}
where $\rho_0$ is the initial network state on $\ox_{k\in \mZ^\nu}\fH_k$, and  $\ups:= \ox_{k\in \mZ^\nu}\ups_k$ is the vacuum state on the Fock space $\fF$, with $\ups_k$ the vacuum states on the corresponding Fock spaces $\fF_k$. The averaging in (\ref{vac}) reduces to that over $\ups$, and the factorizations come from the tensor-product structure of $\fF$, $\ups$ and the commutativity between the quantum Wiener processes $W_k$ on the spaces $\fF_k$ with different $k\in \mZ^\nu$. The state $\rho_0$ in (\ref{rho}) is said to be \emph{proper} if the initial network variables have finite second moments,
and the matrix
\begin{equation}
\label{K}
    K:= (K_{jk})_{j,k\in \mZ^\nu}:= \Re \bE(X(0)X(0)^\rT),
    \qquad
    K_{jk}:= \Re \bE (X_j(0)X_k(0)^\rT),
\end{equation}
acting on $u:=(u_k)_{k \in \mZ^\nu} \in \ell^2(\mZ^\nu, \mR^n)$ as
$    K u
    :=
    \big(
    \sum_{k \in \mZ^\nu}
    K_{jk}u_k
    \big)_{j \in \mZ^\nu}
$,
specifies a bounded operator in the sense of the $\ell^2$-induced norm:
\begin{equation}
\label{Knorm}
    \|K\| < +\infty.
\end{equation}

\begin{theorem}
\label{th:inv}
 Suppose the translation invariant network, described together with related quantities by (\ref{dXj})--(\ref{cABCD}), satisfies the stability condition
\begin{equation}
\label{stab}
    \max_{\sigma \in \mT^\nu}
    \br(\re^{\cA(\sigma)})
    < 1
\end{equation}
(with $\br(\cdot)$ the spectral radius of a matrix),
has a proper  initial state in the sense of (\ref{K}), (\ref{Knorm})
and is driven by the vacuum input fields as specified by (\ref{vac}).
 Then there is weak convergence to a unique invariant Gaussian quantum state for the internal network variables with zero mean and block Toeplitz quantum covariances
\begin{equation}
\label{EXX}
    \bE(X_j(t)X_k(t)^\rT) = P_{j-k} + i\delta_{jk} \Theta,
    \qquad
    j,k\in \mZ^\nu.
\end{equation}
The SFT
\begin{equation}
\label{cP}
      \cP(\sigma)
    :=
    \sum_{\ell\in \mZ^\nu}
    \re^{-i\ell^\rT \sigma}
      P_\ell,
    \qquad
    \sigma \in \mT^\nu,
\end{equation}
for the real parts  $P_\ell = P_{-\ell}^\rT \in \mR^{n\x n}$ of (\ref{EXX})
is found uniquely from the algebraic Lyapunov equation (ALE)
\begin{equation}
\label{cPALE}
  \cA(\sigma) \cP(\sigma)+ \cP(\sigma)\cA(\sigma)^* + \cB(\sigma) \cB(\sigma)^*
   = 0.
\end{equation}
\hfill$\square$
\end{theorem}

The matrix $P$, obtained in (\ref{Elim}) of the proof of the above theorem in \ref{sec:invproof}, can be shown to belong to $\fT_{n,n}$ if the SFTs $\cA$, $\cB$ have an appropriate degree of smoothness.

\begin{lemma}
\label{lem:smooth}
In addition to the assumptions of Theorem~\ref{th:inv}, suppose the SFTs $\cA$, $\cB$ in (\ref{cABCD}) are $r$ times continuously differentiable, with
\begin{equation}
\label{rnu}
  r> \nu.
\end{equation}
Then the block Toeplitz matrix $P$ in (\ref{Elim}) of the invariant real covariances in (\ref{EXX}) belongs to the Banach algebra $\fT_{n,n}$.
\hfill$\square$
\end{lemma}
\begin{proof}
Due to (\ref{vec}), the SFT $\cP$ inherits the $r$ times continuous differentiability from $\cA$, $\cB$.   This implies that the partial derivatives  $\d_{\sigma_k}^r\cP(\sigma)$ with respect to the coordinates of $\sigma:= (\sigma_k)_{1\< k \< \nu}\in \mT^\nu$ are continuous and hence, square integrable over the torus  $\mT^\nu$. Therefore, application of the Plancherel identity yields
\begin{align}
\nonumber
    +\infty
    & >
    \frac{1}{(2\pi)^\nu}
    \int_{\mT^\nu}
    \sum_{k = 1}^\nu
    \|\d_{\sigma_k}^r\cP(\sigma)\|_\rF^2
    \rd \sigma\\
\label{Hold}
    & =
    \sum_{\ell\in \mZ^\nu}
    \|\ell\|_{2r}^{2r}
    \|P_\ell\|_\rF^2
    \>
    \nu^{1-r}
    \sum_{\ell\in \mZ^\nu}
    |\ell|^{2r}
    \|P_\ell\|_\rF^2,
\end{align}
where $\|\cdot\|_{\rF}$ is the Frobenius norm of matrices\cite{HJ_2007}, and use is made of the inequality $\|\ell\|_{2r} := \sqrt[2r]{\sum_{k=1}^\nu \ell_k^{2r}} \> \nu^{\frac{1-r}{2r}}|\ell|$ for a vector $\ell:= (\ell_k)_{1\< k\< \nu}$. It follows from the convergence of the rightmost  series in (\ref{Hold}) that
$
    \|P_\ell\|_\rF = o(|\ell|^{-r})
$, as $|\ell| \to +\infty$, which, in combination with (\ref{rnu}), leads to $\|P\|_1 \<  \sum_{\ell\in \mZ^\nu}
    \|P_\ell\|_\rF < +\infty$, whereby $P \in \fT_{n,n}$.
\end{proof}

In the finite range interaction case, mentioned in Section~\ref{sec:par}, the SFTs $\cA$, $\cB$ in (\ref{cA}), (\ref{cB}) are trigonometric polynomials and are, therefore, infinitely differentiable. Therefore, in this case, (\ref{rnu}) is satisfied,  and it follows from Lemma~\ref{lem:smooth} and its proof that the invariant covariances of the network variables have an infinitely differentiable  SFT $\cP(\sigma)$ whose entries are organised as ratios of trigonometric polynomials of $\sigma \in \mT^\nu$. In the univariate case of $\nu=1$, this makes   $\cP$ have the structure of spectral densities associated with linear discrete-time invariant systems and admit appropriate inner-outer factorizations\cite{W_1972}.

Similarly to Ref. \refcite{VPJ_2018a}, under the conditions of Theorem~\ref{th:inv}, the internal network variables have an invariant multipoint zero-mean  Gaussian quantum state which is specified completely by the two-point quantum covariances:
\begin{equation}
\label{EXX2}
    \bE(X(t)X(\tau)^\rT)
    =
    \left\{
    \begin{matrix}
    \re^{(t-\tau)A}(P + i\bTheta) & {\rm if} &      t \> \tau \> 0 \\
    (P + i\bTheta) \re^{(\tau-t)A^\rT} & {\rm if} &      \tau \> t \> 0
    \end{matrix}
    \right.,
\end{equation}
where $\bTheta$ is given by (\ref{TJJ}). In accordance with the translation invariant structure of the network, (\ref{EXX2}) is also a block Toeplitz   matrix, which, under the conditions of Lemma~\ref{lem:smooth}, is an element of $\fT_{n,n}$.

\section{Finite-Horizon Quadratic-Exponential Functional}
\label{sec:QEF}

Associated with every lattice site $j \in \mZ^\nu$ is a vector $Z_j$ of $q\< n$ time-varying self-adjoint quantum variables,  which represent physical quantities (in regard to the $j$th component system and its neighbourhood)  whose moderate values are preferable for network performance.
These ``critical'' quantum variables are assumed to be linearly related to the internal network variables by a given real block Toeplitz weighting matrix $S:= (S_{j-k})_{j,k\in \mZ^\nu} \in \fT_{q,n}$ and form an auxiliary quantum process
\begin{equation}
\label{ZX}
    Z
    :=
    (Z_j)_{j \in \mZ^\nu}
    :=
    \Big(
        \sum_{k\in \mZ^\nu} S_{j-k} X_k
    \Big)_{j \in  \mZ^\nu}
    =
    SX.
\end{equation}
The matrix $S$ quantifies the relative importance of the network variables in (\ref{ZX}) depending on a particular control application and is not constrained by PR conditions.
Consider a fragment of the network at a nonempty finite subset $G \subset \mZ^\nu$. Similarly to (\ref{XG}), the corresponding restriction
\begin{equation}
\label{ZG}
    Z_G
    :=
    (Z_j)_{j\in G}
    =
    S_G X
\end{equation}
of the process (\ref{ZX}) is related to the network variables by the matrix
\begin{equation}
\label{SG}
    S_G:= (S_{j-k})_{j \in G, k \in \mZ^\nu}
\end{equation}
with $\# G$ rows. In
the risk-sensitive framework, the performance of the network fragment in terms of the process $Z_G$ over a bounded time interval $[0,T]$ can be described by using a quadratic-exponential functional (QEF)\cite{VPJ_2018a}
\begin{equation}
\label{XiGT}
    \Xi_{\theta,G,T}
    :=
    \bE \re^{\theta Q_{G,T} }.
\end{equation}
This cost imposes an exponential penalty (whose severity is controlled by a scalar parameter $\theta >0$)
on the positive semi-definite self-adjoint quantum variable
\begin{align}
\nonumber
    Q_{G,T}
    & :=
    \frac{1}{2}
    \int_0^T
    \sum_{j\in G}
    Z_j(t)^\rT Z_j(t)
    \rd t\\
\label{QGT}
    & =
    \frac{1}{2}
    \int_0^T
    Z_G(t)^\rT Z_G(t)
    \rd t
    =
    \frac{1}{2}
    \int_0^T
    X(t)^\rT
    S_G^\rT  S_GX(t)
    \rd t,
\end{align}
where the integrand is organised  similarly to the Hamiltonian (\ref{HG}). The restricted weighting matrix $S_G$ in (\ref{SG}) specifies the  quadratic dependence of $Q_{G,T}$ on the past history of the network variables. The quantum average of (\ref{QGT}) is related to the asymptotic behaviour of the QEF (\ref{XiGT}) for small values of the risk sensitivity parameter $\theta$ as
\begin{align}
\nonumber
    \bE Q_{G,T}
    & =
    \d_\theta \ln \Xi_{\theta,G,T}\big|_{\theta = 0}\\
\label{EQGT}
    & =
    \frac{1}{2}
    \int_0^T
    \Tr (S_G \Re \bE(X(t)X(t)^\rT) S_G^\rT)
    \rd t.
\end{align}
In what follows, it is assumed that the network satisfies the conditions of Theorem~\ref{th:inv} and is in the invariant multipoint Gaussian quantum state. In this case, the mean square cost functional (\ref{EQGT}) has the following rate per unit time and lattice site:
\begin{align}
\nonumber
    \frac{1}{T\#G}
    \bE Q_{G,T}
    & =
    \frac{1}{2}
    \bE(Z_0(0)^\rT Z_0(0))    \\
\nonumber
    & =
    \frac{1}{2}
    \sum_{j,k\in \mZ^\nu}
    \Tr (S_{-j} P_{j-k}S_{-k}^\rT)\\
    & =
    \frac{1}{2(2\pi)^\nu}
    \int_{\mT^\nu}
    \Tr (\cS(\sigma) \cP(\sigma)\cS(\sigma)^*)
    \rd \sigma,
\label{EQGT1}
\end{align}
where use is made of the Plancherel identity along with the SFT $\cP$ from (\ref{cP}), (\ref{cPALE}) and the SFT for the weighting matrix $S$ in (\ref{ZX}):
\begin{equation}
\label{cS}
      \cS(\sigma)
    :=
    \sum_{\ell\in \mZ^\nu}
    \re^{-i\ell^\rT \sigma}
      S_\ell,
    \qquad
    \sigma \in \mT^\nu.
\end{equation}
The relations (\ref{EQGT}), (\ref{EQGT1}) suggest that similar limits exist for the   infinite spatio-temporal horizon  asymptotic behaviour of the QEF (\ref{XiGT}):
\begin{align}
\label{UpsGdef}
  \Ups_{\theta,G}
  & :=
  \lim_{T\to +\infty}
  \Big(
  \frac{1}{T}
  \ln \Xi_{\theta, G, T}
  \Big),\\
\label{Ups0}
    \Ups(\theta)
    & :=
  \lim_{G\to \infty}
  \Big(
  \frac{1}{\# G}
  \Ups_{\theta,G}
  \Big),
\end{align}
where ``$G\to \infty$'' will be  specified in Section~\ref{sec:spat} and includes, as a particular case, sequences of unboundedly growing cubes in $\mZ^\nu$.

The QEF growth rate (\ref{Ups0}), as a function of $\theta> 0$,
can be used for  large deviations estimates for quantum trajectories of the network in the form of upper bounds on tail probabilities, similar to those in Refs. \refcite{VPJ_2018a,VPJ_2021}.  More precisely, application of an exponential inequality\cite{S_1996} to
the probability distribution\cite{H_2001}  $\bP_{G,T}$ of the self-adjoint quantum variable $Q_{G,T}$ in (\ref{QGT}) leads to
\begin{equation}
\label{Cramer}
    \bP_{G,T}([\eps, +\infty))
    \<
    \inf_{\theta >0}
    (
    \Xi_{\theta, G,T}
    \re^{-\eps \theta }),
    \qquad
    \eps\> 0,
\end{equation}
for any $T>0$ and nonempty finite set $G\subset \mZ^\nu$. By using (\ref{Cramer}) with $\eps = \alpha T\#G$   in combination with (\ref{UpsGdef}), (\ref{Ups0}), it follows that
\begin{equation}
\label{tail}
    \limsup_{G\to \infty}
    \Big(
    \frac{1}{\# G}
    \limsup_{T\to +\infty}
    \Big(
        \frac{1}{T}
        \ln
        \bP_{G,T}([\alpha T \#G, +\infty))
    \Big)
    \Big)
    \<
    \inf_{\theta>0}
    (
        \Ups(\theta)
        -
        \alpha\theta
    )
\end{equation}
for any $\alpha>0$. The relation (\ref{tail}) provides asymptotic upper bounds for the tail probability distribution of $Q_{G,T}$ in terms of the spatio-temporal QEF growth rate (\ref{Ups0}). These bounds can be enhanced by minimizing $\Ups(\theta)$ (at a suitably chosen $\theta>0$) over an admissible range of parameters of the quantum network.  This provides a risk-sensitive performance criterion for  quantum feedback network control by interconnection, exemplified in Fig.~\ref{fig:loop}. The computation of the  bounds (\ref{tail}) and the QEF minimization   require systematic techniques for evaluating the functional (\ref{Ups0}).

In order to establish the existence of and compute the limits (\ref{UpsGdef}), (\ref{Ups0}) in Sections~\ref{sec:temp}--\ref{sec:homo}, we will now discuss the quantum probabilistic structure of the process $Z$ in (\ref{ZX}). The multipoint zero-mean  Gaussian structure of the invariant quantum state of the internal network variables is inherited by the process $Z$ which has the two-point quantum covariances
\begin{align}
\nonumber
    \bE(Z(t)Z(\tau)^\rT)
    & =
    S\bE(X(t)X(\tau)^\rT)    S^\rT\\
\nonumber
    & =
    \left\{
    \begin{matrix}
    S\re^{(t-\tau)A}(P + i\bTheta)S^\rT & {\rm if} &      t \> \tau \> 0 \\
    S(P + i\bTheta) \re^{(\tau-t)A^\rT}S^\rT & {\rm if} &      \tau \> t \> 0
    \end{matrix}
    \right.\\
\label{EZZ2}
    & =
    V(t-\tau) + i\Lambda(t-\tau),
    \qquad
    t,\tau \> 0.
\end{align}
This time-invariant\footnote{that is, depending on the time difference} $\fT_{q,q}$-valued quantum covariance kernel is
obtained by an appropriate transformation of (\ref{EXX2}). Its real part is
given by
\begin{equation}
\label{V}
    V(\tau)
    =
    \left\{
    {\begin{matrix}
    S\re^{\tau A}P S^\rT& {\rm if}\  \tau\> 0\\
    S P\re^{-\tau A^\rT}S^\rT & {\rm if}\  \tau < 0
    \end{matrix}}
    \right.
    =
    V(-\tau)^\rT,
    \qquad
    \tau \in \mR,
\end{equation}
where $P$ is the matrix (\ref{Elim}) of real parts of the invariant one-point quantum covariances of the internal network variables. The imaginary part of (\ref{EZZ2}) is given by
\begin{equation}
\label{Lambda}
    \Lambda(\tau)
     =
    \left\{
    {\begin{matrix}
    S \re^{\tau A}\bTheta S^\rT & {\rm if}\  \tau\> 0\\
    S \bTheta\re^{-\tau A^{\rT}}S^\rT & {\rm if}\  \tau< 0\\
    \end{matrix}}
    \right.
    =
    -\Lambda(-\tau)^\rT,
    \qquad
    \tau \in \mR,
\end{equation}
and describes the two-point CCRs\cite{VPJ_2018a}
\begin{equation}
\label{ZZcomm}
    [Z(t), Z(\tau)^\rT]
    =
    2i\Lambda(t-\tau),
    \qquad
    t,\tau\>0,
\end{equation}
from which the one-point CCR matrix of $Z$  is recovered as $\Lambda(0) = S\bTheta S^\rT$. Accordingly, the process $Z_G$ in (\ref{ZG}) is in a multipoint zero-mean Gaussian state with the time-invariant $\mC^{G\x G}$-valued quantum covariance kernel
\begin{align}
\nonumber
    \bE(Z_G(t)Z_G(\tau)^\rT)
    & =
    S_G\bE(X(t)X(\tau)^\rT)    S_G^\rT\\
\nonumber
    & =
    \left\{
    \begin{matrix}
    S_G\re^{(t-\tau)A}(P + i\bTheta)S_G^\rT & {\rm if} &      t \> \tau \> 0 \\
    S_G(P + i\bTheta) \re^{(\tau-t)A^\rT}S_G^\rT & {\rm if} &      \tau \> t \> 0
    \end{matrix}
    \right.\\
\label{EZZG}
    & =
    V_G(t-\tau) + i\Lambda_G(t-\tau),
    \qquad
    t,\tau \> 0,
\end{align}
which is
obtained as an appropriate restriction of (\ref{EZZ2}) to the set $G\subset \mZ^\nu$  in view of (\ref{SG}) and is split into the real and imaginary parts $V_G$, $\Lambda_G$.  The latter is given by
\begin{equation}
\label{LambdaG}
    \Lambda_G(\tau)
     =
    \left\{
    {\begin{matrix}
    S_G \re^{\tau A}\bTheta S_G^\rT & {\rm if}\  \tau\> 0\\
    S_G \bTheta\re^{-\tau A^{\rT}}S_G^\rT & {\rm if}\  \tau< 0\\
    \end{matrix}}
    \right.
    =
    -\Lambda_G(-\tau)^\rT,
    \qquad
    \tau \in \mR,
\end{equation}
and,  in accordance with (\ref{Lambda}), (\ref{ZZcomm}),  describes the two-point CCRs
\begin{equation}
\label{ZZcommG}
    [Z_G(t), Z_G(\tau)^\rT]
    =
    2i\Lambda_G(t-\tau),
    \qquad
    t,\tau\>0,
\end{equation}
where $\Lambda_G(0) = S_G\bTheta S_G^\rT$ is the one-point CCR matrix of $Z_G$.
The two-point CCR kernel (\ref{LambdaG}) gives rise to a skew self-adjoint integral operator $\sL_{G,T} : f\mapsto g$  which acts on the Hilbert space $L^2([0,T],\mC^G)$ of square integrable $\mC^G$-valued functions on the time interval $[0,T]$ as
\begin{equation}
\label{sLGT}
    g(t)
    :=
    \int_0^T
    \Lambda_G(t-\tau) f(\tau)
    \rd \tau,
    \qquad
    0\< t \< T.
\end{equation}
The commutation structure (\ref{LambdaG}), (\ref{ZZcommG})  of the process $Z_G$, and the related operator $\sL_{G,T}$ in (\ref{sLGT}),  do not depend on a particular network-field state (\ref{rho}).  The real part of the quantum covariance kernel (\ref{EZZG}) is given by
\begin{equation}
\label{VG}
    V_G(\tau)
    =
    \left\{
    {\begin{matrix}
    S_G\re^{\tau A}P S_G^\rT& {\rm if}\  \tau\> 0\\
    S_G P\re^{-\tau A^\rT}S_G^\rT & {\rm if}\  \tau < 0
    \end{matrix}}
    \right.
    =
    V_G(-\tau)^\rT,
    \qquad
    \tau \in \mR,
\end{equation}
in accordance with (\ref{V}).  The kernel $V_G$ specifies a positive semi-definite self-adjoint integral  operator $\sV_{G,T} : f\mapsto g$ acting on $L^2([0,T], \mC^G)$ as
\begin{equation}
\label{sVGT}
    g(t)
    :=
    \int_0^T
    V_G(t-\tau)f(\tau)\rd \tau,
    \qquad
    0\< t \< T.
\end{equation}
The fact that $\sV_{G,T}  \succcurlyeq 0$ also follows from the stronger property of positive semi-definiteness of the self-adjoint operator $\sV_{G,T} +i\sL_{G,T} $ on $L^2([0,T], \mC^G)$. With (\ref{EZZG}) being a continuous kernel, both $\sV_{G,T}$ and $\sL_{G,T}$ are compact operators\cite{RS_1980}. Application of appropriately  modified results of Refs.~\refcite{VPJ_2019c,VPJ_2021} to the quantum process $Z_G$ in the multipoint Gaussian quantum state allows the QEF (\ref{XiGT}) to be  represented as
\begin{equation}
\label{lnXi}
  \ln \Xi_{\theta,G,T}
    =
    -  \frac{1}{2}
  \Tr (\ln\cos (\theta\sL_{G,T} ) + \ln (\cI - \theta \sV_{G,T} \sK_{\theta,G,T}   )).
\end{equation}
Here, $\cI$ is the identity operator on $L^2([0,T],\mC^G)$, and use is made of a positive definite self-adjoint operator
\begin{equation}
\label{sK}
    \sK_{\theta,G,T}
    :=
    \tanhc(i\theta \sL_{G,T} )
    =
    \tanc (\theta\sL_{G,T} ),
\end{equation}
where $\tanhc z := \tanc (-iz)$ is a hyperbolic version of $\tanc z := \frac{\tan z}{z}$ extended  by continuity as $\tanc 0:=1$. The operator $\sK_{\theta,G,T}  $ is nonexpanding in the sense that $\sK_{\theta,G,T}   \preccurlyeq \cI$. With $\sV_{G,T} \sK_{\theta,G,T}$ being a compact operator (which is isospectral to the positive semi-definite self-adjoint operator $\sqrt{\sK_{\theta,G,T}  } \sV_{G,T}  \sqrt{\sK_{\theta,G,T}  }$), the representation (\ref{lnXi}) is valid under the condition
\begin{equation}
\label{spec}
    \theta \lambda_{\max}(\sV_{G,T} \sK_{\theta,G,T}  ) < 1.
\end{equation}
The representation (\ref{lnXi}) is obtained by applying the results of Refs.  \refcite{VPJ_2019c,VPJ_2021} to the Gaussian quantum process $Z_G$ in (\ref{ZG}), (\ref{QGT})   using its quantum Karhunen-Loeve expansion over an orthonormal eigenbasis of the operator $\sL_{G,T} $ in (\ref{sLGT}), provided the latter has no zero eigenvalues.
A sufficient condition for this property to hold for all sufficiently large subsets $G\subset \mZ^\nu$ and time horizons $T>0$  can be developed in terms of the parameters of the quantum network and the weighting  matrix $S$ in (\ref{ZX}) and its SFT (\ref{cS}). However, in the network setting, this development is more complicated than in the case of a single OQHO (see Theorem 10.1 of Ref. \refcite{VPJ_2021}) and requires a separate investigation, which is beyond the scope of the present study and will be discussed elsewhere. In what follows, the absence of  zero eigenvalues will be used as an assumption.

\section{Temporal QEF Growth Rate}\label{sec:temp}

We will first compute the infinite time horizon asymptotic growth rate (\ref{UpsGdef}) of the QEF (\ref{XiGT}) for a fixed but otherwise arbitrary nonempty finite set $G \subset \mZ^\nu$. The dependence on $G$ will be indicated for the subsequent computation of the limit (\ref{Ups0}) in Section~\ref{sec:spat}. As a preliminary for the theorem below, note that the representation (\ref{lnXi}) is organised as ``trace-analytic''\cite{VP_2010} functionals of operators in the sense that
\begin{equation}
\label{lnXi1}
  \ln \Xi_{\theta,G,T}
    =
    -  \frac{1}{2}
  \Tr (\varphi(\theta \sV_{G,T} \sK_{\theta,G,T}   ) + \psi(\theta\sL_{G,T} )),
\end{equation}
where
\begin{equation}
\label{phipsi}
  \varphi(z):= \ln(1-z),
  \qquad
  \psi(z):= \ln \cos z,
  \qquad
  z \in \mC,
\end{equation}
are analytic functions whose domains contain the spectra of the operators $\theta \sV_{G,T} \sK_{\theta,G,T}   $ (under the condition (\ref{spec})) and $\theta\sL_{G,T} $, at which these functions are evaluated. The structure of the operators $\sV_{G,T}$ in (\ref{sVGT}) and $\sL_{G,T} $ in  (\ref{sLGT}) (with the latter giving rise to $\sK_{\theta,G,T}$ in (\ref{sK}))  plays a part together with the averaging relations of  \ref{sec:averint} in the following theorem on the asymptotic behaviour of the quantity (\ref{lnXi1}),  as
$T\to +\infty$, which is an adaptation of Theorem 8.1 of Ref. \refcite{VPJ_2021}.  Its formulation employs the $\mC^{G\x G}$-valued Fourier transforms
\begin{align}
\label{PhiG}
    \Phi_G(\lambda)
    & :=
    \int_\mR \re^{-i\lambda t }
    V_G(t)
    \rd t
    =
    F_G(i\lambda) F_G(i\lambda)^*,\\
\label{PsiG}
    \Psi_G(\lambda)
    & :=
    \int_\mR \re^{-i\lambda t }
    \Lambda_G(t)
    \rd t
    =
    F_G(i\lambda) \bitJ_m F_G(i\lambda)^*,
    \qquad
    \lambda \in \mR,
\end{align}
of the covariance and commutator kernels (\ref{VG}), (\ref{LambdaG}); see also Eq.~(5.8) in Ref. \refcite{VPJ_2019a}. Here,
\begin{equation}
\label{FG}
    F_G(v)
    :=
    S_G
    (v\bitI_n - A)^{-1}B,
    \qquad
    v \in \mC,
\end{equation}
is the $\mC^{G\x \mZ^\nu}$-valued transfer function from the incremented input quantum Wiener process $W$ of the network in (\ref{dX}) to the stationary Gaussian quantum process $Z_G$ in (\ref{ZG}), with      $\bitI_n:= (\delta_{jk}I_n)_{j,k\in \mZ^\nu}$.  Note that $\Phi_G(\lambda)$ is a complex positive semi-definite Hermitian  matrix, while  $\Psi_G(\lambda)$  is skew Hermitian for any $\lambda \in \mR$, with $\Phi_G+i\Psi_G$ being the Fourier transform  of the quantum covariance kernel $V_G+i\Lambda_G$ from (\ref{EZZG}).

\begin{theorem}
\label{th:limXi}
 Suppose the translation invariant network in (\ref{dXj})--(\ref{cABCD}) satisfies the conditions of Theorem~\ref{th:inv}, and the integral operator  $\sL_{G,T}$ in (\ref{sLGT}) has no zero eigenvalues for all sufficiently large $T>0$. Also, let the risk sensitivity parameter $\theta>0$ in (\ref{XiGT}) satisfy
\begin{equation}
\label{spec1}
    \theta
    \sup_{\lambda \in \mR}
    \lambda_{\max}
    (
        \Phi_G(\lambda)
        \tanc
        (\theta \Psi_G(\lambda))
    )
    < 1,
\end{equation}
where the functions $\Phi_G$, $\Psi_G$ are associated with the finite subset $G \subset \mZ^\nu$ by (\ref{PhiG})--(\ref{FG}). Then the QEF $\Xi_{\theta,G,T}$, defined by  (\ref{XiGT}), (\ref{QGT}), has the following infinite time horizon growth rate (\ref{UpsGdef}):
\begin{equation}
\label{UpsG}
    \Ups_{\theta,G}
     =
    -
    \frac{1}{4\pi}
    \int_{\mR}
    \ln\det
    D_{\theta,G}(\lambda)
    \rd \lambda,
\end{equation}
where
\begin{equation}
\label{DG}
    D_{\theta,G}(\lambda)
    :=
    \cos(
        \theta \Psi_G(\lambda)
    ) -
        \theta
        \Phi_G(\lambda)
        \sinc
        (\theta \Psi_G(\lambda))
\end{equation}
is a $\mC^{G\x G}$-valued function,
and $\sinc z := \frac{\sin z}{z}$ (which is extended  by continuity as $\sinc 0 := 1$).\hfill$\square$
\end{theorem}
\begin{proof}
The proof is similar to that of Theorem 8.1 of Ref. \refcite{VPJ_2021} and is outlined for completeness. Since the case of one integral operator is free from noncommutativity, (\ref{lim1}) applies directly to the term $\Tr \psi(\theta\sL_{G,T})$ in (\ref{lnXi1}), with the function $\psi$ given by (\ref{phipsi}):
\begin{align}
\nonumber
    \lim_{T\to +\infty}
    \Big(
        \frac{1}{T}
        \Tr \psi(\theta\sL_{G,T} )
    \Big)
    & =
    \frac{1}{2\pi}
    \int_{\mR}
    \Tr \ln\cos(
        \theta \Psi_G(\lambda)
    )
    \rd \lambda\\
\label{psilim}
    & =
    \frac{1}{2\pi}
    \int_{\mR}
    \ln\det \cos(
        \theta \Psi_G(\lambda)
    )
    \rd \lambda,
\end{align}
where the identity $\Tr \ln N = \ln\det N$ for square matrices $N$ is used along with the  Fourier transform (\ref{PsiG})
of the commutator kernel (\ref{LambdaG}).
Application of (\ref{lim1}) to $\Tr \varphi(\theta \sV_{G,T} \sK_{\theta,G,T}   )$ in (\ref{lnXi1}),  with the function $\varphi$ from (\ref{phipsi}),  involves  two noncommuting integral operators $\sV_{G,T} $, $\sL_{G,T} $ in (\ref{sVGT}), (\ref{sLGT}) and the related operator $\sK_{\theta,G,T}  $ from (\ref{sK}) as
\begin{align}
\nonumber
    \varphi(\theta \sV_{G,T}  \sK_{\theta,G,T}  )
    & =
    -
    \sum_{N=1}^{+\infty}
    \frac{1}{N}
    \theta^N
    (\sV_{G,T}  \sK_{\theta,G,T})^N\\
\label{phiPK}
    & =
    -
    \sum_{N=1}^{+\infty}
    \frac{1}{N}
    \theta^N
    \sum_{k_1, \ldots, k_N = 0}^{+\infty}
    \rprod_{j=1}^N
    \big(
    c_{k_j}
    \theta^{2k_j}
    \sV_{G,T}
    \sL_{G,T} ^{2k_j}
    \big)
\end{align}
under the condition (\ref{spec}).
Here, the Maclaurin series expansion $\tanc z = \sum_{k=0}^{+\infty} c_k z^{2k}$ (with coefficients $c_k \in \mR$) takes into account the symmetry of the {\tt tanc} function.  By applying (\ref{lim1}) to (\ref{phiPK}) in combination with a dominated convergence argument, it follows that
\begin{align}
\nonumber
     \lim_{T\to +\infty}&
    \Big(
        \frac{1}{T}
        \Tr \varphi(\theta\sV_{G,T} \sK_{\theta,G,T}  )
    \Big)\\
\nonumber
    & =
    -
    \frac{1}{2\pi}
    \sum_{N=1}^{+\infty}
    \frac{1}{N}
    \theta^N
    \sum_{k_1, \ldots, k_N = 0}^{+\infty}
    \int_\mR
    \Tr
    \rprod_{j=1}^N
    \big(
    c_{k_j}
    \theta^{2k_j}
    \Phi_G(\lambda)
    \Psi_G(\lambda)^{2k_j}
    \big)
    \rd \lambda\\
\nonumber
    & =
    \frac{1}{2\pi}
    \int_{\mR}
    \Tr
    \ln(
    I_{\#G} -
        \theta
        \Phi_G(\lambda)
        \tanc
        (\theta \Psi_G(\lambda))
    )
    \rd \lambda\\
\label{philim}
    & =
    \frac{1}{2\pi}
    \int_{\mR}
    \ln\det(
    I_{\#G} -
        \theta
        \Phi_G(\lambda)
        \tanc
        (\theta \Psi_G(\lambda))
    )
    \rd \lambda,
\end{align}
where the Fourier transforms (\ref{PhiG}), (\ref{PsiG}) are used.  The limit relation (\ref{philim}) holds under the condition (\ref{spec1}) which is a frequency-domain counterpart of (\ref{spec}).
A combination of (\ref{psilim}), (\ref{philim}) leads to the following  asymptotic growth rate (\ref{UpsGdef}) for the quantity (\ref{lnXi1}):
\begin{align}
\nonumber
    \Ups_{\theta, G}
    =&
    -
    \frac{1}{4\pi}
    \int_{\mR}
    \ln\det(
    I_{\#G} -
        \theta
        \Phi_G(\lambda)
        \tanc
        (\theta \Psi_G(\lambda))
    )
    \rd \lambda\\
\nonumber
    & -
    \frac{1}{4\pi}
    \int_{\mR}
    \ln\det \cos(
        \theta \Psi_G(\lambda)
    )
    \rd \lambda\\
\label{Xilim}
    = &
    -
    \frac{1}{4\pi}
    \int_{\mR}
    \ln\det(
    \cos(
        \theta \Psi_G(\lambda)
    ) -
        \theta
        \Phi_G(\lambda)
        \sinc
        (\theta \Psi_G(\lambda))
    )
    \rd \lambda,
\end{align}
where the identity
$\tanc z \cos z = \sinc z$ is applied to the matrix $\theta \Psi_G(\lambda)$.  In view of (\ref{DG}),  the relation (\ref{Xilim}) is identical to (\ref{UpsG}). %\hfill$\blacksquare$
\end{proof}

Under the condition (\ref{spec1}), the quantity $-\ln\det D_{\theta,G}(\lambda)$ is a nonnegative-valued  symmetric function of the frequency $\lambda\in \mR$. This symmetry allows the integration in (\ref{UpsG}) to be reduced as
$    \Ups_{\theta,G}
     =
    -
    \frac{1}{2\pi}
    \int_0^{+\infty}
    \ln\det
    D_{\theta,G}(\lambda)
    \rd \lambda
$.

\section{Spatio-Temporal Growth Rate of the QEF}
\label{sec:spat}

We will now proceed to the spatio-temporal growth rate (\ref{Ups0}) of the QEF (\ref{XiGT}). In view of (\ref{DG}), the representation (\ref{UpsG}) of the temporal QEF growth rate also has a trace-analytic structure
\begin{equation}
\label{UpsG1}
    \Ups_{\theta,G}
    =
    -
    \frac{1}{4\pi}
    \int_{\mR}
    \Tr
    \big(
    \varphi(
        \theta
        \Phi_G(\lambda)
        \tanc
        (\theta \Psi_G(\lambda))
        )
        +
        \psi(\theta \Psi_G(\lambda))
    \big)
    \rd \lambda,
\end{equation}
involving the analytic functions (\ref{phipsi}) along with the $\mC^{G\x G}$-valued functions $\Phi_G$, $\Psi_G$  from (\ref{PhiG}), (\ref{PsiG}). At any given frequency $\lambda\in \mR$,  each of the matrices $\Phi_G(\lambda)$, $\Psi_G(\lambda)$ is organised as the restriction $f_G:= (f_{j-k})_{j,k\in G}$ of a complex  block Toeplitz matrix $f:= (f_{j-k})_{j,k\in \mZ^\nu} \in \fT_{n,n}$ to $G \subset \mZ^\nu$. This will be combined with the averaging relations of \ref{sec:aver} in the theorem below on the asymptotic behavior of (\ref{UpsG}) for ``large'' fragments of the network. More precisely, a nonempty finite set $G \subset \mZ^\nu$ is said to tend to infinity ($G\to \infty$)  if
\begin{equation}
\label{Ginf}
    \Delta_G(\ell)
    :=
    \frac{\#(G \setminus (G+\ell))}{\#G}
  \to 0,
  \qquad
  \ell \in \mZ^\nu.
\end{equation}
The function $\Delta_G: \mZ^\nu\to [0,1]$ is symmetric (that is, $\Delta_G(\ell) = \Delta_G(-\ell)$ for all $\ell\in \mZ^\nu$) and quantifies the relative discrepancy between the set $G$ and its translations $G + \ell = \{z+\ell: z\in G\}$, so that
$$
    \frac{\#(G\Delta (G+\ell))}{\# G}
    =
    2\Delta_G(\ell),
    \qquad
    \frac{\#(G\bigcap(G+\ell))}{\# G} = 1- \Delta_G(\ell),
$$
where $\alpha \Delta \beta $ denotes the symmetric  difference of sets $\alpha$, $\beta$. Accordingly, $\Delta_G(\ell)<1$ holds if and only if $\ell \in G-G:= \{x-y: x,y\in G\}$. Also note that $\sum_{\ell \in \mZ^\nu} (1- \Delta_G(\ell)) = \#G$, whereby (\ref{Ginf}) implies that $\#G \to +\infty$. The latter property is not only necessary but is also  sufficient for $G \to \infty$ in certain classes of sets $G$. In particular, for a cube $G:= \{0,\ldots, L-1\}^\nu$, which consists of $\#G = L^\nu$ lattice sites, where $L$ is a positive integer, the left-hand side of  (\ref{Ginf}) takes the form $\Delta_G(\ell) = 1- \prod_{k=1}^\nu \max(0,1-|\ell_k|/L)$ for any $\ell:=(\ell_k)_{1\< k\< \nu} \in \mZ^\nu$. In this case, the condition $G\to \infty$ in the sense of (\ref{Ginf}) reduces to the side length of the cube unboundedly growing: $L\to +\infty$. Returning to (\ref{Ginf}) in the general case (when $G$ is not necessarily a cube), we note that the convergence $G\to \infty$  is metrizable in the sense of its equivalence to
\begin{equation}
\label{Ginfmet}
    \sum_{\ell \in \mZ^\nu}
    2^{-|\ell_1|-\ldots-|\ell_\nu|}
    \Delta_G(\ell)\to 0.
\end{equation}
The following theorem, which is concerned with the asymptotic behaviour  of the quantity (\ref{UpsG}),   as $G\to \infty$, employs the $\mC^{q\x q}$-valued spatio-temporal Fourier transforms
\begin{align}
\label{Phi}
    \Phi(\sigma,\lambda)
    & :=
    \sum_{\ell \in \mZ^\nu}
    \int_{\mR}
    \re^{-i(\ell^\rT \sigma + \lambda t )}
    V_\ell(t)
    \rd t
    =
    F(\sigma,i\lambda)
    F(\sigma,i\lambda)^*,\\
\label{Psi}
    \Psi(\sigma,\lambda)
    & :=
    \sum_{\ell \in \mZ^\nu}
    \int_{\mR}
    \re^{-i(\ell^\rT \sigma + \lambda t )}
    \Lambda_\ell(t)
    \rd t
    =
    F(\sigma,i\lambda)
    J_m
    F(\sigma,i\lambda)^*,
    \qquad
    \sigma\in \mT^\nu,\
    \lambda \in \mR,
\end{align}
of the invariant two-point covariance and commutator kernels of the process $Z$ in (\ref{ZX}). Here,
\begin{equation}
\label{F}
    F(\sigma,s)
    :=
    \cS(\sigma)(sI_n - \cA(\sigma))^{-1}\cB(\sigma),
    \qquad
    \sigma \in \mT^\nu,\
    s \in \mC,
\end{equation}
is the spatio-temporal transfer function from the incremented input fields of the network to $Z$. Similarly to (\ref{PhiG}), (\ref{PsiG}), $\Phi(\sigma,\lambda)$ is a complex positive semi-definite Hermitian  matrix, while  $\Psi(\sigma,\lambda)$  is skew Hermitian for any $\sigma \in \mT^\nu$, $\lambda \in \mR$, and $\Phi+i\Psi$ is the Fourier transform  of the quantum covariance kernel $V+i\Lambda$ from (\ref{EZZ2}). The function $\Phi+i\Psi: \mT^\nu \x \mR \to \mC^{q\x q}$ can be interpreted as a ``quantum spectral density'' of the process $Z$.

\begin{theorem}
\label{th:limXiG}
 Suppose the translation invariant network in (\ref{dXj})--(\ref{cABCD}) satisfies the conditions of Theorem~\ref{th:inv}, and the integral operator  $\sL_{G,T}$ in (\ref{sLGT}) has no zero eigenvalues for all sufficiently large $T>0$ and finite sets $G\subset \mZ^\nu$ in the sense of (\ref{Ginf}) (or (\ref{Ginfmet})). Also, let the risk sensitivity parameter $\theta>0$ in (\ref{XiGT}) satisfy
\begin{equation}
\label{rad}
    \theta
    \sup_{\sigma \in \mT^\nu,\, \lambda \in \mR}
    \lambda_{\max}
    (
        \Phi(\sigma,\lambda)
        \tanc
        (\theta \Psi(\sigma,\lambda))
    )
    < 1,
\end{equation}
where the functions $\Phi$, $\Psi$ are given by (\ref{Phi}), (\ref{Psi}). Then the QEF $\Xi_{\theta,G,T}$, defined by  (\ref{XiGT}), (\ref{QGT}), has the following spatio-temporal growth rate (\ref{Ups0}):
\begin{equation}
\label{Ups}
    \Ups(\theta)
     =
    -
    \frac{1}{2(2\pi)^{\nu+1}}
    \int_{\mT^\nu\x \mR}
    \ln\det
    D_{\theta}(\sigma,\lambda)
    \rd \sigma\rd \lambda,
\end{equation}
where the function $D_\theta: \mT^\nu\x \mR \to \mC^{q\x q}$ is given by
\begin{equation}
\label{D}
    D_{\theta}(\sigma,\lambda)
    :=
    \cos(
        \theta \Psi(\sigma,\lambda)
    ) -
        \theta
        \Phi(\sigma,\lambda)
        \sinc
        (\theta \Psi(\sigma,\lambda)).
\end{equation}
\hfill$\square$
\end{theorem}
\begin{proof}
The proof is similar to that of Theorem~\ref{th:limXi} except that the averaging relations of \ref{sec:aver} are used here instead of \ref{sec:averint} and are applied to the integrands in (\ref{UpsG}) pointwise at every frequency $\lambda$,  which is followed by a dominated convergence argument.  Application of (\ref{limhG}) to the second integrand in  (\ref{UpsG}) yields
\begin{align}
\nonumber
    \lim_{G\to \infty}
    \Big(
        \frac{1}{\#G}
        \Tr \psi(\theta \Psi_G(\lambda))
    \Big)
    & =
    \frac{1}{(2\pi)^\nu}
    \int_{\mT^\nu}
    \Tr \ln\cos(
        \theta \Psi(\sigma,\lambda)
    )
    \rd \sigma\\
\label{psilimG}
    & =
    \frac{1}{(2\pi)^\nu}
    \int_{\mT^\nu}
    \ln\det \cos(
        \theta \Psi(\sigma,\lambda)
    )
    \rd \sigma,
    \qquad
    \lambda \in \mR,
\end{align}
where use is made of the function $\psi$ from (\ref{phipsi}) and the  Fourier transform (\ref{Psi}) of the commutator kernel (\ref{Lambda}). Application of (\ref{limhG}) to the first integrand in (\ref{UpsG1}) leads to
\begin{align}
\nonumber
     \lim_{G\to \infty}&
    \Big(
        \frac{1}{\#G}
        \Tr \varphi(\theta\Phi_G(\lambda)\tanc(\theta\Psi_G(\lambda)))
    \Big)\\
\nonumber
    & =
    \frac{1}{(2\pi)^\nu}
    \int_{\mT^\nu}
    \Tr
    \ln(
    I_q -
        \theta
        \Phi(\sigma,\lambda)
        \tanc
        (\theta \Psi(\sigma,,\lambda))
    )
    \rd \sigma\\
\label{philimG}
    & =
    \frac{1}{(2\pi)^\nu}
    \int_{\mT^\nu}
    \ln\det(
    I_q -
        \theta
        \Phi(\sigma,\lambda)
        \tanc
        (\theta \Psi(\sigma,\lambda))
    )
    \rd \sigma,
    \qquad
    \lambda \in \mR,
\end{align}
where
the Fourier transform (\ref{Phi}) of the real covariance kernel (\ref{V}) is used together with (\ref{Psi}).  The limit (\ref{philimG}) holds under the condition (\ref{rad}) which is a spatio-temporal frequency-domain counterpart of (\ref{spec1}).
By combining (\ref{psilimG}), (\ref{philimG}), it follows that the quantity (\ref{UpsG1}) has the following  asymptotic growth rate (\ref{Ups0}):
\begin{align}
\nonumber
    \Ups(\theta)
    =&
    -
    \frac{1}{2(2\pi)^{\nu+1}}
    \int_{\mT^\nu\x \mR}
    \ln\det(
    I_q -
        \theta
        \Phi(\sigma,\lambda)
        \tanc
        (\theta \Psi(\sigma,\lambda))
    )
    \rd \sigma\rd \lambda\\
\nonumber
    & -
    \frac{1}{2(2\pi)^{\nu+1}}
    \int_{\mT^\nu\x \mR}
    \ln\det \cos(
        \theta \Psi(\sigma,\lambda)
    )
    \rd \sigma\rd \lambda\\
\label{XilimG}
    = &
    -
    \frac{1}{2(2\pi)^{\nu+1}}
    \int_{\mT^\nu\x \mR}
    \ln\det(
    \cos(
        \theta \Psi(\sigma,\lambda)
    ) -
        \theta
        \Phi(\sigma,\lambda)
        \sinc
        (\theta \Psi(\sigma,\lambda))
    )
    \rd \sigma\rd \lambda.
\end{align}
In view of (\ref{D}), the relation (\ref{XilimG}) establishes (\ref{Ups}). %\hfill$\blacksquare$
\end{proof}

Consider Theorem~\ref{th:limXiG} in the limiting classical case obtained formally by letting $\Theta = 0$ in (\ref{Xcomm}) and $J_m = 0$ in (\ref{Wcomm}). In this case, (\ref{dX}) is an SDE driven by independent standard Wiener processes $W_k$ with values in $\mR^m$ at lattice sites $k\in \mZ^\nu$. The classical invariant measure of the network makes $Z$ in (\ref{ZX}) a stationary $(\mR^q)^{\mZ^\nu}$-valued Gaussian random process\cite{GS_2004} with zero mean and the spectral density $\Phi$ in (\ref{Phi}). Accordingly, the function $\Psi$ in (\ref{Psi}) vanishes, and the condition (\ref{rad}) takes the form
\begin{equation}
\label{class}
    \theta
    <
    \theta_*
    :=
    \frac{1}{\sup_{\sigma \in \mT^\nu,\, \lambda \in \mR}
    \lambda_{\max}
    (
        \Phi(\sigma, \lambda)
    )}
    =
    \frac{1}{\|F\|_\infty^2},
\end{equation}
involving the spatio-temporal counterpart
$$
    \|F\|_\infty
    :=
    \sup_{\sigma \in \mT^\nu,\, \lambda \in \mR}\|F(\sigma, i\lambda)\|
$$
of the Hardy space $\cH_\infty$-norm  for  the transfer function $F$ in  (\ref{F}) which factorizes the spectral density $\Phi$ in (\ref{Phi}). In this case, the right-hand side of (\ref{Ups}) reduces to
\begin{equation}
\label{Ups*}
    \Ups_*(\theta)
    :=
    -
    \frac{1}{2(2\pi)^{\nu+1}}
    \int_{\mT^\nu \x \mR}
    \ln\det(
        I_q
        -
        \theta
        \Phi(\sigma,\lambda)
    )
    \rd \sigma
    \rd \lambda
\end{equation}
in view of (\ref{D}) and corresponds to the $\cH_\infty$-entropy integral of Ref.~\refcite{MG_1990} (see also Ref.~\refcite{AK_1981}).

In contrast to its classical counterpart (\ref{Ups*}), the QEF growth rate (\ref{Ups}) in the quantum case depends  on both functions $\Phi$, $\Psi$ which constitute the quantum spectral density $\Phi + i\Psi$ of the process $Z$ in (\ref{ZX}). Furthermore, the condition (\ref{rad}) is transcendental in $\theta$ and, unlike (\ref{class}),  does not admit a closed-form representation. However, since {\tt tanc} on the imaginary axis (that is, {\tt tanhc} on the real axis) takes values in the interval $(0,1]$ and hence, $0\prec  \tanc (\theta \Psi)\preccurlyeq I_q $,  then
$$
    \lambda_{\max}
    (
        \Phi
        \tanc
        (\theta \Psi)
    )
    =
    \lambda_{\max}
    \big(
        \sqrt{\tanc
        (\theta \Psi)}
        \Phi
        \sqrt{\tanc
        (\theta \Psi)}
    \big)
    \< \lambda_{\max}(\Phi)
$$
everywhere in $\mT^\nu\x \mR$, so that the fulfillment of the classical constraint (\ref{class}) secures (\ref{rad}).

\section{A Homotopy Technique for Computing the QEF Growth Rate}
\label{sec:homo}

Consider the computation of the QEF growth rate (\ref{Ups}) by a technique,  which
resembles the homotopy methods for numerical solution of parameter dependent algebraic equations\cite{MB_1985} and exploits the specific dependence of $\Ups(\theta)$ on the risk sensitivity parameter $\theta$. With the function $D_\theta$ in (\ref{D}),  we associate a function $U_\theta: \mT^\nu\x \mR \to \mC^{q\x q}$ by
\begin{equation}
\label{U}
  U_\theta(\sigma,\lambda):= -D_\theta(\sigma,\lambda)^{-1}\d_\theta D_\theta(\sigma, \lambda)
\end{equation}
for all $\theta >0$ satisfying (\ref{rad}) (which ensures that $\det D_\theta (\sigma,\lambda)\ne 0$ for all $\sigma \in \mT^\nu$, $\lambda \in \mR$). The following theorem provides a network counterpart of Theorem 9.1 from Ref. \refcite{VPJ_2021} (the latter corresponds formally to the single OQHO case with $\nu=0$).

\begin{theorem}
\label{th:diff}
Under the conditions of Theorem~\ref{th:limXiG}, the QEF growth rate $\Ups(\theta)$ in  (\ref{Ups}) satisfies the differential equation
\begin{equation}
\label{Ups'}
  \Ups'(\theta)
  =
    \frac{1}{2(2\pi)^{\nu+1}}
    \int_{\mT^\nu\x \mR}
    \Tr U_\theta(\sigma,\lambda)
    \rd \sigma\rd \lambda,
\end{equation}
with the initial condition $\Ups(0)=0$. Here, the function (\ref{U})  is computed as
\begin{equation}
\label{U1}
U_\theta
=
\Psi
    (    \Psi\cos(
        \theta \Psi
        ) -
        \Phi
        \sin
        (\theta \Psi)
        )^{-1}
        (\Phi \cos(\theta \Psi)
    +\Psi\sin(\theta \Psi)
    )
\end{equation}
(the arguments  $\sigma$, $\lambda$ are omitted for brevity),
takes values in the subspace of Hermitian matrices of order $n$ and satisfies a Riccati equation
\begin{equation}
\label{U'}
  \d_\theta U_\theta(\sigma,\lambda)
  =
  \Psi(\sigma,\lambda)^2
  +
  U_\theta(\sigma,\lambda)^2,
  \qquad
  \sigma \in \mT^\nu,\
  \lambda \in \mR,
\end{equation}
with the initial condition $U_0 = \Phi$ given by (\ref{Phi}).\hfill$\square$
\end{theorem}
\begin{proof}
The relation (\ref{Ups'}) is obtained by combining (\ref{Ups}) with $(\ln\det D_\theta)' = -\Tr U_\theta$, which follows from (\ref{U})  and the identity $(\ln\det N)' = \Tr (N^{-1}N')$, where $(\cdot)':= \d_\theta(\cdot)$.
Since the function $D_\theta$ in (\ref{D})  admits the representation
\begin{equation}
\label{D1}
    D_\theta
    =
    \cos(
        \theta \Psi
    ) -
        \Phi
        \Psi^{-1}
        \sin
        (\theta \Psi)
\end{equation}
for any $\sigma \in \mT^\nu$, $\lambda\in \mR$ (with $\frac{\sin(\theta z)}{z}$ extended by continuity to $\theta$ at $z=0$), its  derivative with respect to $\theta$ takes the form
\begin{equation}
\label{D'}
    D_\theta'
    =
    -\Psi
    \sin(
        \theta \Psi
    )
    -
    \Phi
        \cos
        (\theta \Psi).
\end{equation}
The equality (\ref{U1}) results from
substitution of (\ref{D1}), (\ref{D'}) into (\ref{U}). By differentiating  (\ref{D'}) in $\theta$, it follows that (\ref{D1})
satisfies the linear second-order  ODE
\begin{equation}
\label{D''}
  D_\theta''
  =
    -\Psi^2
    \cos(
        \theta \Psi
    )
    +
    \Phi
    \Psi
        \sin
        (\theta \Psi)
    =
  - D_\theta \Psi^2,
\end{equation}
with the initial conditions $D_0 = I_q$, $D_0' = -\Phi$. In view of the relation $(N^{-1})' = -N^{-1}N'N^{-1}$,  the differentiation of (\ref{U}) leads to
\begin{equation}
\label{U'1}
    U_\theta'
    =
    -D_\theta^{-1}D_\theta''
    +D_\theta^{-1}D_\theta' D_\theta^{-1}D_\theta'
    =
    \Psi^2 + U_\theta^2,
\end{equation}
which uses (\ref{D''}) and establishes (\ref{U'}). The solution $U_\theta$ of this differential equation inherits the Hermitian property from its initial condition $U_0 = \Phi$, since  $\Psi(\sigma,\lambda) = -\Psi(\sigma,\lambda)^*$ in (\ref{Psi}) for any $\sigma \in \mT^\nu$, $\lambda\in \mR$, and $(N^2)^* = N^2$ for Hermitian or skew  Hermitian matrices $N$.
\end{proof}

The transformation $D_\theta\mapsto U_\theta$ in (\ref{U}), which involves a matrix-valued counterpart  of the logarithmic derivative and relates the quadratically nonlinear Riccati ODE (\ref{U'}) to the linear ODE (\ref{D''}), resembles the Hopf-Cole transformation\cite{C_1951,H_1950} linking the viscous Burgers equation with the heat equation.  The role of (\ref{U}) in (\ref{U'1}) is also similar to that of the logarithmic transformation in dynamic programming
equations for stochastic control\cite{F_1982} (see also Ref. \refcite{VP_2010}).

The right-hand side of (\ref{Ups'}) can be evaluated by numerical integration over the spatio-temporal frequencies
and used for computing (\ref{Ups}) as
$$
    \Ups(\theta)
     =
    \int_0^\theta \Ups'(\vartheta)\rd \vartheta
=
    \frac{1}{2(2\pi)^{\nu+1}}
    \int_{\mT^\nu\x \mR \x [0,\theta]}
    \Tr U_\vartheta(\sigma,\lambda)
    \rd \sigma \rd \lambda \rd \vartheta.
$$
In particular, (\ref{Ups'}) yields
\begin{align}
\nonumber
    \Ups'(0)
    & =
    \frac{1}{2(2\pi)^{\nu+1}}
    \int_{\mT^\nu\x\mR}
    \Tr \Phi(\sigma,\lambda)
    \rd \sigma
    \rd \lambda\\
\label{Ups'0}
    & =
    \frac{1}{2}
    \|F\|_2^2 =
    \frac{1}{2}
    \bE(Z_0(0)^\rT Z_0(0)),
\end{align}
which, in accordance with (\ref{EQGT}), (\ref{EQGT1}),  reproduces the mean square cost rate for the process $Z$ in (\ref{ZX}) in the invariant Gaussian state of the network. In (\ref{Ups'0}), use is also made of a spatio-temporal version
$$
    \|F\|_2
    :=
    \sqrt{
    \frac{1}{(2\pi)^{\nu+1}}
    \int_{\mT^\nu\x\mR}
    \|F(\sigma,i\lambda)\|_\rF^2
    \rd \sigma
    \rd \lambda}
$$
of the Hardy space $\cH_2$-norm  for  the transfer function $F$ in  (\ref{F}) which factorizes $\Phi$ in (\ref{Phi}).

The function $\det (I_q - \theta \Phi(\sigma,\lambda))$ in the classical QEF rate (\ref{Ups*}) is rational with respect to $\lambda$, simplifying the evaluation of the integral. This observation can be combined with the Maclaurin series expansions of the trigonometric functions, which allows (\ref{D}) to be approximated as
\begin{align}
\nonumber
    D_\theta
    & =
    I_q - \frac{1}{2}\theta^2 \Psi^2 - \theta \Phi \Big(I_q - \frac{1}{6}\theta^2 \Psi^2\Big)
    +
    o(\theta^3)\\
\label{Dprox}
    & =
    I_q - \theta \Phi
    -
    \frac{1}{2}\theta^2
    \Big(I_q - \frac{\theta}{3}\Phi\Big)\Psi^2
    +
    o(\theta^3),
        \qquad
    {\rm as}\
    \theta \to 0.
\end{align}
Substitution of (\ref{Dprox}) into (\ref{Ups}) allows the quantum QEF growth rate to be computed approximately through a perturbation of its classical counterpart (\ref{Ups*}):
\begin{align}
\nonumber
  \Ups(\theta)
  = &
  \Ups_*(\theta)\\
\nonumber
  & +
  \frac{\theta^2}{4(2\pi)^{\nu+1}}
  \int_{\mT^\nu \x\mR}
  \Tr
  \Big((I_q - \theta\Phi(\sigma,\lambda))^{-1}\Big(I_q - \frac{\theta}{3}\Phi(\sigma,\lambda)\Big)\Psi(\sigma,\lambda)^2
  \Big)
  \rd \sigma
  \rd \lambda\\
\label{VUps}
    & + o(\theta^3),
    \qquad
    {\rm as}\
    \theta \to 0.
\end{align}
Since $\Psi(\sigma,\lambda)^2\prec 0$ for all $\sigma \in\mT^\nu$, $\lambda \in \mR$, the relation (\ref{VUps}) implies that $\Ups(\theta)< \Ups_*(\theta)$ for all sufficiently small $\theta >0$.

\section{Conclusion}
\label{sec:conc}

We have considered a class of translation invariant networks of multimode OQHOs on a multidimensional lattice, governed by linear QSDEs driven by external quantum fields. The block Toeplitz structure of their  coefficients has been exploited in order to represent the PR conditions in the spatio-temporal frequency domain, relate them with the energy and coupling matrices, and compute the energy parameters for interconnections of networks. Such interconnections  arise in quantum control settings with network  performance specifications including stability and minimization of cost functionals. We have discussed  the invariant Gaussian quantum state for stable networks, driven by vacuum fields,  and a quadratic-exponential cost functional as a risk-sensitive performance criterion for finite fragments of the network over bounded time intervals. This cost gives rise to exponential upper bounds for tail distributions of a quadratic function of network variables weighted by a block Toeplitz matrix. A spatio-temporal frequency-domain formula has been obtained for the asymptotic QEF rate per unit time and per lattice site in the thermodynamic limit of infinite time horizons and unboundedly growing network fragments. This representation involves the quantum spectral density,  associated through the Fourier transform with the invariant quantum covariance kernel of the network variables and factorised by the  spatio-temporal transfer function of the network. We have obtained a  differential equation for the QEF rate as a function of the risk sensitivity parameter and outlined its computation using a homotopy technique and asymptotic expansions. These  results provide a solution of the risk-sensitive performance analysis problem  in the spatio-temporal frequency domain for translation invariant linear quantum stochastic networks, which can be applied to coherent and measurement-based control and filtering settings for such systems with QEF criteria.

\appendix

\section{Block Toeplitz Matrices and Spatial Fourier Transforms}
\label{sec:Toeplitz}

Omitting the dependence on the dimension $\nu$ of the lattice $\mZ^\nu$, which is fixed throughout the paper, we denote by $\fT_{a,b}$ the Banach space of real or complex block Toeplitz matrices $f:= (f_{j-k})_{j,k\in \mZ^\nu}$ (in the sense of the additive group  structure of the lattice $\mZ^\nu$)
 with $(a\x b)$-blocks $f_{j-k}$, endowed with the maximum absolute row (or column) sum norm\cite{HJ_2007}
\begin{equation}
\label{f1}
    \|f\|_1
    :=
    \sup_{j \in \mZ^\nu}
    \sum_{k \in \mZ^\nu}\|f_{j-k}\|
    =
    \sum_{\ell \in \mZ^\nu}\|f_\ell\|
    =
    \sup_{k \in \mZ^\nu}
    \sum_{j \in \mZ^\nu}\|f_{j-k}\|,
\end{equation}
where $\|\cdot\|$ is the operator norm of a matrix.
The condition $\|f\|_1< +\infty$
makes the spatial Fourier transform (SFT)
\begin{equation}
\label{phi}
    F(\sigma): =\sum_{k \in \mZ^\nu}\re^{-ik^\rT \sigma}f_k
\end{equation}
a continuous function of $\sigma \in \mT^\nu$ ($2\pi$-periodic in each if its $\nu$ variables)     in view of the absolute  and uniform summability of the series over the $\nu$-dimensional torus $\mT^\nu$,  with $\mT$ being identified with the interval $[-\pi, \pi)$ in what follows. The torus $\mT^\nu$ is a commutative group with respect to the entrywise addition modulo $2\pi$. The matrix $f \in \fT_{a,b}$ specifies a bounded operator $\ell^2(\mZ^\nu, \mC^b)\to \ell^2(\mZ^\nu, \mC^a)$ for the Hilbert spaces of square summable complex vector-valued  functions on the lattice $\mZ^\nu$.
 The corresponding $\ell^2$-induced operator norm of $f$ satisfies $\|f\|= \max_{\sigma \in \mT^\nu} \|F(\sigma)\| \< \|f\|_1$, where
$\|F(\sigma)\|$ inherits continuous dependence on $\sigma\in \mT^\nu$ from $F(\sigma)$.
The complex conjugate transpose $(\cdot)^* := \overline{(\cdot)}^\rT$ maps $f \in \fT_{a,b}$ to $f^* = (f_{k-j}^*)_{j,k\in \mZ^\nu} \in \fT_{b,a}$ with the SFT $F(\sigma)^*$.
 The product of matrices $f:= (f_{j-k})_{j,k\in \mZ^\nu} \in \fT_{a,b}$ and $g:= (g_{j-k})_{j,k\in \mZ^\nu}\in \fT_{b,c}$ is also a block Toeplitz matrix $h:= fg=(h_{j-k})_{j,k\in \mZ^\nu} \in \fT_{a,c}$ whose blocks are given by the convolutions $h_j = \sum_{k\in \mZ^\nu} f_{j-k}g_k= \sum_{k\in \mZ^\nu} f_kg_{j-k}$ for all $j \in \mZ^\nu$, with $\|h\|_1 \< \|f\|_1\|g\|_1$ in view of the submultiplicativity of the matrix operator norm which is used on the right-hand side of (\ref{f1}).  The corresponding SFT $H(\sigma): =\sum_{k \in \mZ^\nu}\re^{-ik^\rT \sigma}h_k = F(\sigma)G(\sigma)$ is the product of (\ref{phi}) and $G(\sigma): =\sum_{k \in \mZ^\nu}\re^{-ik^\rT \sigma}g_k$ for all $\sigma \in \mT^\nu$. Accordingly,
 $\fT_{a,a}$ is a Banach algebra of block Toeplitz matrices whose multiplication corresponds to the pointwise multiplication  of the SFTs. For any $f \in \fT_{a,a}$, its exponential is also a block Toeplitz matrix $\re^f \in \fT_{a,a}$ which satisfies $\|\re^f\|_1 \< \re^{\|f\|_1}$, and the corresponding SFT is related to (\ref{phi}) by $\re^{F(\sigma)}$.

\section{Proof of Theorem~\ref{th:PR}}
\label{sec:PRproof}

By using the bilinearity of commutators and applying the quantum Ito lemma,  it follows from (\ref{dXj}) that
\begin{align}
\nonumber
    \rd [X_j, X_k^\rT]
    = &
    [\rd  X_j, X_k^\rT]
    +
    [X_j, \rd   X_k^\rT]
    +
    [\rd  X_j, \rd  X_k^\rT]\\
\nonumber
    = &
    \sum_{a \in \mZ^\nu}\Big[ A_{j-a}X_a \rd t + B_{j-a} \rd W_a, X_k^\rT\Big]\\
\nonumber
    & +
    \sum_{b \in \mZ^\nu} \Big[X_j, X_b^\rT A_{k-b}^\rT \rd t + \rd W_b^\rT B_{k-b}^\rT \Big]\\
\nonumber
    &+
    \sum_{a,b \in \mZ^\nu}
    \Big[
    A_{j-a}X_a \rd t + B_{j-a} \rd W_a, X_b^\rT A_{k-b}^\rT \rd t + \rd W_b^\rT B_{k-b}^\rT
    \Big]\\
\label{dXcomm}
    =&
    \sum_{c\in \mZ^\nu}
    (A_{j-c}
    [X_c,X_k^\rT]
    +
    [X_j, X_c^\rT] A_{k-c}^\rT
    +
    2i B_{j-c}J_m B_{k-c}^\rT
    )
    \rd t
\end{align}
for all $j,k\in \mZ^\nu$,
where (\ref{dWcomm}) and the second of the equalities (\ref{WXYdWcomm}) are also used. It follows from (\ref{dXcomm}) that the preservation of the CCRs (\ref{Xcomm}) is equivalent to
\begin{equation}
\label{PR11}
    A_\ell
    \Theta
    +
    \Theta
    A_{-\ell}^\rT
    +
    \sum_{c\in \mZ^\nu}
    B_{\ell + c}J_m B_c^\rT
    =0,
    \qquad
    \ell\in \mZ^\nu,
\end{equation}
since $    \sum_{c\in \mZ^\nu}
    B_{j-c}J_m B_{k-c}^\rT =     \sum_{c\in \mZ^\nu}
    B_{j-k+c}J_m B_c^\rT
$ for all $j,k\in \mZ^\nu$. The first PR condition (\ref{PR1}) is obtained by applying the SFT to (\ref{PR11})  and using (\ref{cAB*}). Now,  in view of (\ref{dX}),
for any $t\> s\> 0$,
\begin{equation}
\label{Xsol}
    X(t) = \re^{(t-s)A} X(s) + \int_s^t \re^{(t-\tau) A}B \rd W(\tau),
\end{equation}
where the integral consists of quantum variables which commute with adapted processes taken at time $s$ (see also (\ref{WXYdWcomm})). Hence,
\begin{equation}
\label{XYXY}
    [X_j(t), Y_k(s)^\rT] = \sum_{\ell\in \mZ^\nu }(\re^{(t-s)A})_{j\ell} [X_\ell(s), Y_k(s)^\rT],
\end{equation}
where $(\re^{\tau A})_{j\ell}$ is the $(j,\ell)$th block of the matrix $\re^{\tau A} \in \fT_{n,n}$ satisfying $\|\re^{\tau A}\|_1 \< \re^{\tau \|A\|_1}$ for any $\tau \> 0$. The relation (\ref{XYXY}) shows that (\ref{XYcomm}) holds if and only if it does so for all $s=t\> 0$.
By considering the processes $X_j$, $Y_k$ at the same moment of time $t\> 0$, it follows from the quantum Ito lemma and (\ref{dWcomm})--(\ref{dYj}), similarly to (\ref{dXcomm}), that
\begin{align}
\nonumber
    \rd [X_j, Y_k^\rT]
     =&
    [\rd  X_j, Y_k^\rT]
    +
    [X_j, \rd   Y_k^\rT]
    +
    [\rd  X_j, \rd  Y_k^\rT]\\
\nonumber
    = &
    \sum_{a \in \mZ^\nu}\Big[ A_{j-a}X_a \rd t + B_{j-a} \rd W_a, Y_k^\rT\Big]\\
\nonumber
    & +
    \sum_{b \in \mZ^\nu} \Big[X_j, X_b^\rT C_{k-b}^\rT \rd t + \rd W_b^\rT D_{k-b}^\rT \Big]\\
\nonumber
    &+
    \sum_{a,b \in \mZ^\nu}
    \Big[
    A_{j-a}X_a \rd t + B_{j-a} \rd W_a, X_b^\rT C_{k-b}^\rT \rd t + \rd W_b^\rT D_{k-b}^\rT
    \Big]\\
\label{dXYcomm}
    =&
    \sum_{c\in \mZ^\nu}
    (A_{j-c}
    [X_c,Y_k^\rT]
    +
    [X_j, X_c^\rT] C_{k-c}^\rT
    +
    2i B_{j-c}J_m D_{k-c}^\rT
    )
    \rd t
\end{align}
for all $j,k\in \mZ^\nu$.
A combination of (\ref{Xcomm}), (\ref{XYcomm}) with (\ref{dXYcomm}) shows that, under the CCRs (\ref{Xcomm}),  the preservation of (\ref{XYcomm}) by the QSDEs (\ref{dXj}), (\ref{dYj}) is equivalent to
\begin{equation}
\label{PR21}
    \Theta C_{-\ell}^\rT
    +
    \sum_{c\in \mZ^\nu}
    B_{\ell + c}J_m D_c^\rT
    =0,
    \qquad
    \ell \in \mZ^\nu.
\end{equation}
The second PR condition (\ref{PR2}) is now obtained by applying the SFT to (\ref{PR21}) and using (\ref{cCD*}). By a reasoning, similar to that in (\ref{XYXY}), a combination of (\ref{Xsol}) with (\ref{XYcomm}) yields
\begin{align}
\nonumber
    [Y_j(t), Y_k(s)^\rT]
    & =
    \Big[
    Y_j(s) + \int_s^t \sum_{\ell \in \mZ^\nu} (C_{j-\ell}X_\ell(\tau)\rd \tau + D_{j-\ell} \rd W_\ell(\tau)),
    Y_k(s)^\rT
    \Big]\\
\nonumber
    & = [Y_j(s),Y_k(s)^\rT]
    +
    \sum_{\ell \in \mZ^\nu}
    C_{j-\ell}
    \int_s^t
    [X_\ell(\tau), Y_k(s)^\rT]
    \rd \tau\\
\label{YjYkcomm}
    &
        =
        [Y_j(s),Y_k(s)^\rT]
\end{align}
for all $j,k\in \mZ^\nu$, $t\> s\> 0$. The relation (\ref{YjYkcomm}) implies that (\ref{Ycomm}) is valid if and only if it holds for all $s=t\> 0$.
By considering the processes $Y_j$, $Y_k$ at the same moment of time $t\> 0$ and combining the quantum Ito lemma with (\ref{dWcomm})--(\ref{dYj}) similarly to (\ref{dXYcomm}), it follows that
\begin{align}
\nonumber
    \rd [Y_j, Y_k^\rT]
    = &
    [\rd  Y_j, Y_k^\rT]
    +
    [Y_j, \rd   Y_k^\rT]
    +
    [\rd  Y_j, \rd  Y_k^\rT]\\
\nonumber
    = &
    \sum_{a \in \mZ^\nu}\Big[ C_{j-a}X_a \rd t + D_{j-a} \rd W_a, Y_k^\rT\Big]\\
\nonumber
    & +
    \sum_{b \in \mZ^\nu} \Big[Y_j, X_b^\rT C_{k-b}^\rT \rd t + \rd W_b^\rT D_{k-b}^\rT \Big]\\
\nonumber
    &+
    \sum_{a,b \in \mZ^\nu}
    \Big[
    C_{j-a}X_a \rd t + D_{j-a} \rd W_a, X_b^\rT C_{k-b}^\rT \rd t + \rd W_b^\rT D_{k-b}^\rT
    \Big]\\
\nonumber
    =&
    \sum_{c\in \mZ^\nu}
    (C_{j-c}
    [X_c,Y_k^\rT]
    +
    [Y_j, X_c^\rT] C_{k-c}^\rT
    +
    2i D_{j-c}J_m D_{k-c}^\rT
    )
    \rd t\\
\label{dYYcomm}
    =&
    2i
    \sum_{c\in \mZ^\nu}
    D_{j-c}J_m D_{k-c}^\rT
    \rd t
\end{align}
for all $j,k\in \mZ^\nu$, where use is also made of the CCRs (\ref{XYcomm}). Therefore, (\ref{dYYcomm}) reproduces the incremental form $  \rd [Y_j, Y_k^\rT]
  =
    [\rd Y_j, \rd Y_k^\rT]
  =
  2i\delta_{jk}
  J_r\rd t
$ of (\ref{Ycomm}) if and only if     $\sum_{c\in \mZ^\nu}
D_{\ell+c}J_m D_{c}^\rT = \delta_{\ell 0} J_r$ for all $\ell \in \mZ^\nu$, which is equivalent to (\ref{PR3}) obtained through the SFT.\hfill$\blacksquare$

\section{Proof of Theorem~\ref{th:inv}}
\label{sec:invproof}

For any fixed but otherwise arbitrary  $u:=(u_k)_{k \in \mZ^\nu} \in \ell^2(\mZ^\nu, \mR^n)$, consider the QCF of the internal network variables  at time $t\> 0$, defined by averaging their unitary Weyl operator\cite{F_1989}:
\begin{equation}
\label{chi}
    \phi(t,u)
    :=
    \bE \re^{i u^\rT X(t)}
    =
    \bE \prod_{k \in \mZ^\nu}
    \re^{i u_k^\rT X_k(t)},
\end{equation}
where the factorisation  comes from the commutativity $[X_j(t),X_k(t)^\rT] = 0$ for different  sites $j\ne k$ of the lattice in view of (\ref{Xcomm}). Similarly to Lemma~1 of Ref. \refcite{VPJ_2018a},
a combination of (\ref{Xsol}) with (\ref{vac})  leads to
\begin{equation}
\label{phit}
    \phi(t,u)
    =
    \phi(0,\re^{tA^\rT}u)
    \re^{-\frac{1}{2}\|u\|_{E(t)}^2},
\end{equation}
where $\|u\|_{E}:= \sqrt{u^\rT E u}$ is a weighted Euclidean norm in $\ell^2(\mZ^\nu, \mR^n)$, specified by a time-varying real positive semi-definite symmetric block Toeplitz matrix
\begin{equation}
\label{Et}
    E(t) := (E_{j-k}(t))_{j,k\in \mZ^\nu} =  \int_0^t \re^{\tau A}BB^\rT \re^{\tau A^\rT} \rd \tau \in \fT_{n,n}
\end{equation}
with the SFT
\begin{equation}
\label{cE}
    \cE_t(\sigma)
     :=
    \sum_{\ell \in \mZ^\nu}
    \re^{-i\ell^\rT \sigma}
    E_\ell(t)=
    \int_0^t
    \re^{\tau\cA(\sigma)}
    \cB(\sigma)\cB(\sigma)^*
    \re^{\tau\cA(\sigma)^*} \rd \tau,
    \qquad
    \sigma \in \mT^\nu.
\end{equation}
Since the SFT  $\cA$ is continuous over the torus $\mT^\nu$,
the condition (\ref{stab}) is equivalent to the matrix $\cA(\sigma)$ being Hurwitz for any $\sigma\in \mT^\nu$ and ensures that (\ref{cE}) has a pointwise limit
\begin{equation}
\label{cElim}
    \cP(\sigma)
    :=
    \lim_{t\to +\infty}
    \cE_t(\sigma)
     =
    \int_0^{+\infty}
    \re^{\tau\cA(\sigma)}
    \cB(\sigma)\cB(\sigma)^*
    \re^{\tau\cA(\sigma)^*} \rd \tau,
    \qquad
    \sigma \in \mT^\nu,
\end{equation}
with the convergence being monotonic in the sense that $\cP(\sigma)\succcurlyeq \cE_t(\sigma)\succcurlyeq \cE_\tau(\sigma)$ for all $t\> \tau\>0$. The matrix $\cP(\sigma)$ in (\ref{cElim})
is a unique solution of the ALE (\ref{cPALE}) and inherits continuity in $\sigma \in \mT^\nu$ from  $\cA$, $\cB$ due to (\ref{stab}) and the representation
\begin{equation}
\label{vec}
\col
    (\cP(\sigma)) = -(\cA(-\sigma) \op \cA(\sigma))^{-1}(\cB(-\sigma)\ox I_n) \col(\cB(\sigma)).
\end{equation}
Here, $\col(\cdot)$ is the columnwise vectorization of matrices\cite{M_1988,SIG_1998},
$\alpha \op \beta = \alpha \ox I + I\ox \beta $ is the Kronecker  sum of matrices $\alpha$, $\beta$, and the relations (\ref{cAB*}) are used.  The matrix $\cA(-\sigma) \op \cA(\sigma)$ in (\ref{vec}) is also Hurwitz (and hence, nonsingular) due to (\ref{stab}). The continuity of the function $\cP$ ensures its square integrability over the torus $\mT^\nu$, thus making $\cP$  a legitimate SFT with square summable Fourier coefficients
\begin{equation}
\label{Pk}
    P_\ell := \frac{1}{(2\pi)^\nu}\int_{\mT^\nu} \re^{i\ell^\rT \sigma} \cP(\sigma)\rd \sigma,
    \qquad
    \ell \in \mZ^\nu,
\end{equation}
which are real matrices due to the Hermitian property,  also inherited by $\cP$ (from $\cA$, $\cB$) as a unique solution of the ALE (\ref{cPALE}): $\overline{\cP(\sigma)} = \cP(-\sigma)$ for all $\sigma \in \mT^\nu$.  The matrices (\ref{Pk}) form a block Toeplitz matrix
\begin{equation}
\label{Elim}
    P :=
    (P_{j-k})_{j,k\in \mZ^\nu}
    =
    \lim_{t\to +\infty} E(t)
    =
    \int_0^{+\infty} \re^{\tau A}BB^\rT \re^{\tau A^\rT} \rd \tau,
\end{equation}
which is the limit of (\ref{Et}) and satisfies $AP + PA^\rT + BB^\rT=0$ whose spatial frequency domain representation is (\ref{cPALE}).  The convergence (\ref{Elim}) is also monotonic: $P \succcurlyeq E(t)\succcurlyeq E(\tau)$ for all $t\> \tau\>0$. This leads to the limit
\begin{equation}
\label{lim2}
  \lim_{t\to +\infty}
  \re^{-\frac{1}{2}\|u\|_{E(t)}^2}
  =
  \re^{-\frac{1}{2}\|u\|_P^2}
\end{equation}
for the second factor on the right-hand side of (\ref{phit}).
Concerning the asymptotic behaviour of the first factor in (\ref{phit}), note that
\begin{align}
\nonumber
    |\phi(0,v)-1|^2
    & = |\bE (\re^{iv^\rT X(0)}-\cI_{\fH})|^2\\
\nonumber
    & \<
    \bE ((\re^{-iv^\rT X(0)}-\cI_{\fH})(\re^{iv^\rT X(0)}-\cI_{\fH}))\\
\nonumber
    & =
    4 \bE ((\sin (v^\rT X(0)/2))^2) \\
\nonumber
    & \<
     \bE ((v^\rT X(0))^2)\\
\label{phidev}
    &
    =
     \|v\|_K^2\< \|K\| |v|^2,
     \qquad
     v \in \ell^2(\mZ^\nu, \mR^n).
\end{align}
Here, the inequality $|\bE \xi|^2 \< \bE (\xi^\dagger\xi)$ for any  quantum variable $\xi$ (with $(\cdot)^\dagger$ the operator adjoint) is applied to $\xi:= \re^{iv^\rT X(0)}-\cI_{\fH}$ and combined with the unitarity of the Weyl operator $\re^{iv^\rT X(0)}$. Also, the inequality $(\sin \eta )^2 \preccurlyeq \eta^2$ for any self-adjoint operator $\eta$ is applied to $\eta:= v^\rT X(0)$ together with (\ref{K}), (\ref{Knorm}). In accordance with (\ref{phit}), the inequality (\ref{phidev}) will subsequently be considered at $v:=\re^{tA^\rT} u$ whose norm satisfies
\begin{equation}
\label{u2}
    |\re^{tA^\rT} u|^2
    =
    \frac{1}{(2\pi)^\nu}
    \int_{\mT^\nu}
    |\re^{t\cA(\sigma)^*}
    \cU(\sigma)|^2
    \rd \sigma,
\end{equation}
where use is made of the Plancherel identity, and
\begin{equation}
\label{cU}
    \cU(\sigma)
    := \sum_{k \in \mZ^\nu} \re^{-ik^\rT \sigma} u_k
\end{equation}
is the SFT of $u$ (the series is convergent in the Hilbert space $L^2(\mT^\nu, \mC^n)$ of square integrable $\mC^n$-valued functions on the torus $\mT^\nu$).
Now, let $\veps>0$ be any positive real number which is small enough in the sense that
\begin{equation}
\label{veps}
    \veps \max_{\sigma \in \mT^\nu}\lambda_{\max}(\cA(\sigma)+\cA(\sigma)^*) \< 1,
\end{equation}
where the continuity of the dependence of the largest eigenvalue $\lambda_{\max}(\cdot)$ on a Hermitian matrix is used.  The condition  (\ref{veps}) is equivalent to
\begin{equation}
\label{gamma}
    \gamma(\sigma)
    :=
    I_n - \veps(\cA(\sigma)+\cA(\sigma)^*)\succcurlyeq 0 ,
    \qquad
    \sigma \in \mT^\nu.
\end{equation}
Hence, the matrix
\begin{align}
\nonumber
    \Gamma(\sigma)
     & :=
    \veps I_n
    +
    \int_0^{+\infty}
    \re^{t\cA(\sigma)}
    \re^{t\cA(\sigma)^*}
    \rd t\\
\label{Gam}
    & =
    \int_0^{+\infty}
    \re^{t\cA(\sigma)}
    \gamma(\sigma)
    \re^{t\cA(\sigma)^*}
    \rd t
    =
    \Gamma(\sigma)^*
\end{align}
satisfies the ALE
\begin{equation}
\label{GALE}
    \cA(\sigma)\Gamma(\sigma) + \Gamma(\sigma)\cA(\sigma)^* + \gamma(\sigma)=0
\end{equation}
and is a continuous Hermitian function of $\sigma\in \mT^\nu$. It follows from (\ref{Gam}) that $\Gamma(\sigma)$ is separated from zero as
\begin{equation}
\label{mineps}
    \min_{\sigma \in \mT^\nu} \lambda_{\min}(\Gamma(\sigma))\> \veps,
\end{equation}
where use is also made of the continuous dependence of the smallest eigenvalue $\lambda_{\min}(\cdot)$ on a Hermitian matrix. A combination of (\ref{GALE}) with (\ref{gamma}) leads to
\begin{align*}
        \d_t(\re^{t\cA(\sigma)}
        \Gamma(\sigma)
    \re^{t\cA(\sigma)^*})
    & =
    \re^{t\cA(\sigma)}
    (\cA(\sigma)\Gamma(\sigma) + \Gamma(\sigma)\cA(\sigma)^*)
    \re^{t\cA(\sigma)^*}\\
    & =
    -
    \re^{t\cA(\sigma)}
    \gamma(\sigma)
    \re^{t\cA(\sigma)^*}
    \preccurlyeq 0 ,
\end{align*}
and hence,
\begin{equation}
\label{upper}
    \re^{t\cA(\sigma)}
        \Gamma(\sigma)
    \re^{t\cA(\sigma)^*}
    \preccurlyeq
    \Gamma(\sigma),
    \qquad
    \sigma \in \mT^\nu,\
    t\> 0.
\end{equation}
It follows from (\ref{mineps}), (\ref{upper}), that the integrand in (\ref{u2}) is bounded above uniformly in $t\> 0$ by a time-independent function:
\begin{align*}
    |\re^{t\cA(\sigma)^*}
    \cU(\sigma)|^2
    & =
    \cU(\sigma)^*
    \re^{t\cA(\sigma)}
    \re^{t\cA(\sigma)^*}
    \cU(\sigma)\\
    & \<
    \frac{1}{\veps}
    \cU(\sigma)^*
    \re^{t\cA(\sigma)}
    \Gamma(\sigma)
    \re^{t\cA(\sigma)^*}
    \cU(\sigma)    \\
    & \<
    \frac{1}{\veps}
    \cU(\sigma)^*
    \Gamma(\sigma)
    \cU(\sigma)   ,
    \qquad
    \sigma \in \mT^\nu,
\end{align*}
which is integrable since
\begin{align*}
    \int_{\mT^\nu}
        \cU(\sigma)^*
    \Gamma(\sigma)
    \cU(\sigma)
    \rd \sigma
    & \<
    \max_{s\in \mT^\nu} \lambda_{\max}(\Gamma(s))
    \int_{\mT^\nu}
    |\cU(\sigma)|^2
    \rd \sigma\\
    & =
    (2\pi)^\nu
    |u|^2
    \max_{s\in \mT^\nu} \lambda_{\max}(\Gamma(s))
    <+\infty
\end{align*}
in view of the continuity of $\Gamma$ on the torus $\mT^\nu$ combined with the Plancherel identity for $u\in \ell^2(\mZ^\nu, \mR^n)$ and the SFT (\ref{cU}). The pointwise convergence $\lim_{t\to +\infty} |\re^{t\cA(\sigma)^*}\cU(\sigma)| = 0$ for any $\sigma \in \mT^\nu$ due to (\ref{stab}) and application of the Lebesgue dominated convergence theorem to (\ref{u2}) yield
$$
    \lim_{t\to +\infty}
    |\re^{tA^\rT} u|
    =
    0,
$$
which, in combination with (\ref{phidev}), implies that $|\phi(0,\re^{tA^\rT}u)-1|
\< \sqrt{\|K\|} |\re^{tA^\rT}u|\to 0$ as $t\to +\infty$, and hence,
\begin{equation*}
\label{limphi1}
    \lim_{t\to +\infty} \phi(0,\re^{tA^\rT}u) = 1.
\end{equation*}
By combining this convergence with (\ref{lim2}),
it follows from  (\ref{phit}) that the QCF (\ref{chi}) is pointwise convergent:
\begin{equation}
\label{limphi}
    \lim_{t\to +\infty}
    \phi(t,u)
    =
    \re^{-\frac{1}{2}\|u\|_P^2},
        \qquad
    u \in \ell^2(\mZ^\nu, \mR^n).
\end{equation}
The right-hand side of (\ref{limphi}) is the QCF of the zero-mean Gaussian quantum state with the real covariances $\Re \bE(X_jX_k^\rT) = P_{j-k}$ in (\ref{EXX}), and hence, there holds weak\cite{B_1968} convergence to this invariant state for the network variables.   The imaginary part $\Im \bE(X_jX_k^\rT) = \frac{1}{2} \Im \bE [X_j, X_k^{\rT}] = \delta_{jk}\Theta$ in (\ref{EXX}) comes  from the CCRs (\ref{Xcomm}),  which are preserved over the course of time regardless of the quantum state.\hfill$\blacksquare$

\section{Averaging for Trace-Analytic Functionals of Integral Operators}
\label{sec:averint}

Each of the operators $\sV_{G,T}$ in (\ref{sVGT}) and $\sL_{G,T} $ in  (\ref{sLGT})
is organised as an integral operator $\sF_{G,T}$ on $L^2([0,T],\mC^G)$ whose kernel $F_{G,T}: [0,T]^2\to \mC^{G\x G}$ is obtained from an absolutely integrable  continuous function $f_G: \mR\to \mC^{G \x G}$  as $F_{G,T}(t,\tau):= f_G(t-\tau)$ for all  $0\< t,\tau\< T$.
The ``rightward'' product $\sF_{G,T} := \rprod_{k=1}^N \sF_{G,T}^{(k)}$ of any number $N$ of such operators  with the kernel functions $F_{G,T}^{(k)}: [0,T]^2 \to \mC^{G\x G}$, generated by absolutely integrable  continuous functions $f_{G}^{(k)}: \mR \to \mC^{G\x G}$ as above, $k=1, \ldots, N$, is an integral operator whose kernel  is an appropriately constrained convolution  $F_{G,T}(t,\tau):= \int_{[0,T]^{N-1}} f_G^{(1)}(t-\tau_1) f_G^{(2)}(\tau_1-\tau_2)\x\ldots \x f_G^{(N)}(\tau_{N-1}-\tau)   \rd \tau_1\x \ldots \x \rd \tau_{N-1}$ for all $0\< t,\tau\< T$. The trace of this operator can be computed as\cite{B_1988,RS_1980}
\begin{equation}
\label{TrFT}
    \Tr \sF_{G,T}
    =
    \int_0^T
    \Tr F_{G,T}(t,t)
    \rd t
    =
    \int_{[0,T]^N}
    \Tr
    \rprod_{k=1}^N
    f_G^{(k)} (t_k - t_{k+1})
    \rd t_1\x \ldots \x \rd t_N,
\end{equation}
where $t_{N+1}:= t_1$.
Application of Lemma~6 from Appendix~C of Ref. \refcite{VPJ_2018a} to (\ref{TrFT}) leads to
\begin{equation}
\label{lim}
    \lim_{T\to +\infty}
    \Big(
        \frac{1}{T}
        \Tr \sF_{G,T}
    \Big)
    =
    \frac{1}{2\pi}
    \int_{\mR}
    \Tr
    \rprod_{k=1}^N
    \Phi_G^{(k)}(\lambda)
    \rd \lambda,
\end{equation}
where $\Phi_G^{(k)}(\lambda):= \int_\mR \re^{-i\lambda t }f_G^{(k)}(t)\rd t$ is the Fourier transform of the kernel function $f_G^{(k)}$. In turn, (\ref{lim}) extends from monomials to holomorphic functions $h$ of $N$ complex variables\cite{H_1990} evaluated at the integral operators $\sF_{G,T}^{(k)}$:
\begin{equation}
\label{lim1}
    \lim_{T\to +\infty}
    \Big(
        \frac{1}{T}
        \Tr h
        \big(
            \sF_{G,T}^{(1)}, \ldots, \sF_{G,T}^{(N)}
        \big)
    \Big)
    =
    \frac{1}{2\pi}
    \int_{\mR}
    \Tr
    h(
    \Phi_G^{(1)}(\lambda),
    \ldots,
    \Phi_G^{(N)}(\lambda)
    )
    \rd \lambda,
\end{equation}
provided both sides of (\ref{lim1}) use the same extension of $h$ to noncommutative variables (such extensions are, in general, not unique).

\section{An Averaging Lemma for Block Toeplitz Matrices}
\label{sec:aver}

The following lemma and its corollary are used for computing the infinite spatio-temporal horizon growth rate in Theorem~\ref{th:limXiG}.

\begin{lemma}
\label{lem:vanhove}
For any $N=1,2,3,\ldots$ and any complex block Toeplitz matrices $f^{(s)}:= (f_{j-k}^{(s)})_{j,k\in \mZ^\nu} \in \fT_{n,n}$ with $\mC^{q\x q}$-valued SFTs
\begin{equation}
\label{cFs}
    \cF_s(\sigma)
    :=
    \sum_{\ell \in \mZ^\nu}
    \re^{-i\ell^\rT \sigma}
    f_\ell^{(s)},
    \qquad
    \sigma \in \mT^\nu, \
    s = 1, \ldots, N,
\end{equation}
the following averaging relation holds for the restrictions $f_G^{(s)}:= (f_{j-k}^{(s)})_{j,k\in G} \in \mC^{G\x G}$:
\begin{equation}
\label{limG}
    \lim_{G\to \infty}
    \Big(
        \frac{1}{\#G}
        \Tr
        \rprod_{s=1}^N
        f_G^{(s)}
    \Big)
    =
    \frac{1}{(2\pi)^\nu}
    \int_{\mT^\nu}
    \Tr
        \rprod_{s=1}^N
        \cF_s(\sigma)
    \rd \sigma,
\end{equation}
where the limit is in the sense of (\ref{Ginf}).\hfill$\square$
\end{lemma}
\begin{proof}
If $N=1$, then (\ref{limG}) reduces to the identity $\Tr f_0^{(1)} =     \frac{1}{(2\pi)^\nu} \int_{\mT^\nu} \Tr \cF_1(\sigma) \rd \sigma$ which follows from the SFT inversion applied to  (\ref{cFs}). Now, assuming that $N>1$,
\begin{align}
\nonumber
        \frac{1}{\#G}
        \Tr
        \rprod_{s=1}^N
        f_G^{(s)}
        & =
        \frac{1}{\#G}
        \sum_{k_1, \ldots, k_N \in G,\ k_{N+1}=k_1}
        \Tr
        \rprod_{s=1}^N
        f_{k_s-k_{s+1}}^{(s)}        \\
\label{lim3}
        & =
        \sum_{z_1, \ldots, z_N \in \mZ^\nu:\, z_1 + \ldots + z_N = 0 }
        h_{G,N}(z_1, \ldots, z_{N-1})
        \Tr
        \rprod_{s=1}^N
        f_{z_s}^{(s)}.
\end{align}
Here,
$$
    h_{G,N}(z_1, \ldots, z_{N-1})
    :=
    \frac{1}{\#G}
    \#
    \left(
    G
    \bigcap
    \bigcap_{s=1}^{N-1}
    \Big(
        G + \sum_{k=1}^s z_k
    \Big)
    \right)
    \in [0,1]
$$
for all $    z_1, \ldots, z_{N-1}\in \mZ^\nu$
admits the bound
$$
    1- h_{G,N}(z_1, \ldots, z_{N-1})
     =
        \frac{1}{\#G}
    \#
    \bigcup_{s=1}^{N-1}
    \Big(
    G\setminus
    \Big(
        G + \sum_{k=1}^s z_k
    \Big)
    \Big)\<
    \sum_{s=1}^{N-1}
    \Delta_G
    \Big(
        \sum_{k=1}^s z_k
    \Big)
$$
(which becomes an equality at $N=2$) in terms of (\ref{Ginf}), whereby
\begin{equation}
\label{hlim}
    \lim_{G\to \infty}
    h_{G,N}(z_1, \ldots, z_{N-1}) = 1,
    \qquad
    z_1, \ldots, z_{N-1}\in \mZ^\nu.
\end{equation}
Since
\begin{equation}
\label{abs}
        \sum_{z_1, \ldots, z_N \in \mZ^\nu}
        \Big|
        \Tr
        \rprod_{s=1}^N
        f_{z_s}^{(s)}
        \Big|
        \<
        n
        \rprod_{s=1}^N
        \|f^{(s)}\|_1<+\infty
\end{equation}
in view of (\ref{f1}),
then, by the Lebesgue dominated convergence theorem, the relation (\ref{hlim}) leads to the following limit for (\ref{lim3}):
\begin{align}
\nonumber
    \lim_{G\to \infty}
    \Big(
        \frac{1}{\#G}
        \Tr
        \rprod_{s=1}^N
        f_G^{(s)}
    \Big)
        & =
        \sum_{z_1, \ldots, z_N \in \mZ^\nu:\, z_1 + \ldots + z_N = 0 }
        \Tr
        \rprod_{s=1}^N
        f_{z_s}^{(s)}\\
\label{lim4}
    & =
        \frac{1}{(2\pi)^\nu}
    \int_{\mT^\nu}
    \Tr
        \rprod_{s=1}^N
        \cF_s(\sigma)
    \rd \sigma,
\end{align}
which establishes (\ref{limG}).  The  last equality in (\ref{lim4}) follows from the identity
\begin{align*}
    \int_{\mT^\nu}
        \rprod_{s=1}^N
        \cF_s(\sigma)
    \rd \sigma
    & =
    \sum_{z_1, \ldots, z_N \in \mZ^\nu}
    \int_{\mT^\nu}
    \re^{-i(z_1 + \ldots + z_N)^\rT \sigma}
    \rd \sigma
        \rprod_{s=1}^N
        f_{z_s}^{(s)}\\
        & =
        (2\pi)^\nu
                \sum_{z_1, \ldots, z_N \in \mZ^\nu:\, z_1 + \ldots + z_N = 0 }\
        \rprod_{s=1}^N
        f_{z_s}^{(s)},
\end{align*}
whose right-hand side is an absolutely summable series by the same reasoning as in (\ref{abs}).
\end{proof}

Similarly to (\ref{lim1}), the relation (\ref{limG}) extends from monomials to holomorphic functions $h$ of $N$ complex variables evaluated at the matrices $f_G^{(k)}$:
\begin{equation}
\label{limhG}
    \lim_{G\to \infty}
    \Big(
        \frac{1}{\#G}
        \Tr h
        \big(
            f_G^{(1)}, \ldots, f_{G}^{(N)}
        \big)
    \Big)
    =
    \frac{1}{(2\pi)^\nu}
    \int_{\mT^\nu}
    \Tr
    h(
    \cF_1(\sigma),
    \ldots,
    \cF_N(\sigma)
    )
    \rd \sigma,
\end{equation}
where both sides use the same extension of $h$ to noncommutative variables.

\section*{Acknowledgement}

This work is supported by the Australian Research Council grant DP210101938.

%\section*{CRediT author statement}
%IGV: Conceptualization, Methodology, Investigation, Formal analysis, Writing --  Original %Draft, Review and Editing. IRP: Funding acquisition.

%Conceptualization 	Ideas; formulation or evolution of overarching research goals and aims
%Methodology 	Development or design of methodology; creation of models
%Formal analysis 	Application of statistical, mathematical, computational, or other formal %techniques to analyze or synthesize study data
%Investigation 	Conducting a research and investigation process, specifically performing the %experiments, or data/evidence collection
%Writing - Original Draft 	Preparation, creation and/or presentation of the published work, %specifically writing the initial draft (including substantive translation)
%Writing - Review & Editing 	Preparation, creation and/or presentation of the published %work by those from the original research group, specifically critical review, commentary or %revision – including pre-or postpublication stages
%Funding acquisition

\end{document}